\documentclass[aps,a4paper,showkeys,nofootinbib,longbibliography,notitlepage,superscriptaddress,twocolumn]{revtex4-2}
\usepackage[utf8]{inputenc}
\usepackage{filecontents}
\usepackage{natbib}
\usepackage{amsmath,amssymb,bm,amsthm}
\usepackage{subfigure}
\usepackage{graphicx}% Include figure files
\usepackage{dcolumn}% Align table columns on decimal point
\usepackage{bm}% bold math
\usepackage[mathlines]{lineno}% Enable numbering of text and display math
\usepackage{booktabs}
\usepackage{color}
\newcounter{one}
\usepackage{url}

\usepackage{textcase}
\usepackage{mathrsfs}

\usepackage{colortbl}
\usepackage{tabularx}
\usepackage{verbatim}
\usepackage{multirow}

\usepackage[T1]{fontenc} 
\usepackage{lmodern}
\usepackage{bbm}
\usepackage[utf8]{inputenc}
\usepackage{amsfonts}
\usepackage{array}

\usepackage{bbm}
\usepackage{enumitem}
\usepackage{umoline}
\usepackage[usenames,svgnames]{xcolor}
\usepackage{natbib}
\usepackage[hyperindex,breaklinks]{hyperref}
\hypersetup{
     colorlinks=true,       		% false: boxed links; true: colored links
     linkcolor=Navy,          	% color of internal links
     citecolor=Navy,            % color of links to bibliography
     filecolor=Navy,      		% color of file links
     urlcolor=Navy,           	% color of external links
    runcolor=cyan,
 }
\usepackage{graphicx}
\usepackage{amsfonts}
\usepackage{amssymb}
\usepackage{amsmath}
\usepackage{bbm}
\usepackage{enumitem}

\usepackage{xcolor}

\usepackage{dsfont}

\newcommand{\tr}[0]{ {\rm tr}}

\newcommand{\half}[1]{{ \rm h}}
\newcommand{\Oorderof}{\mathcal{O}}
\newcommand{\orderof}[1]{\Oorderof(#1)}

\newcommand{\for}[0]{\quad \textrm{for} \quad}

\newcommand{\dist}{d}

\newcommand{\co}{{\rm c}}

\newcommand{\cl}{{\rm cl}}

\newcommand{\Cor}{{\rm Cor}}

\newcommand{\Supp}{{\rm Supp}}

\def\beq{\begin{equation}}
\def\eeq{\end{equation}}
\def\nbeq{\begin{equation*}}
\def\neeq{\end{equation*}}
\def\<{\langle}
\def\>{\rangle}

\def\tr{{\rm tr}}

\newtheorem{theorem}{Theorem}

\newtheorem{lemma}{Lemma}

\newcommand{\br}[1]{\left( #1 \right)}

 \newcommand{\norm}[1]{\left \|  #1 \right \|}

\newcommand{\abs}[1]{\left | #1 \right |}

\usepackage{mathtools}
\def\multiset#1#2{\ensuremath{\left(\kern-.3em\left(\genfrac{}{}{0pt}{}{#1}{#2}\right)\kern-.3em\right)}}

\renewcommand\thefootnote{*\arabic{footnote}}
\begin{document}

%\twocolumn[ \begin{@twocolumnfalse}

%\title{Supplementary Material for  ``Optimal speed limit and provably efficient quantum simulation of interacting bosons ''}
\title{Thermal Area Law in Long-Range Interacting Systems}

\author{Donghoon Kim}
\email{donghoon.kim@riken.jp}
\affiliation{
	Analytical Quantum Complexity RIKEN Hakubi Research Team, RIKEN Center for Quantum Computing (RQC), Wako, Saitama 351-0198, Japan
}

\author{Tomotaka Kuwahara}
\email{tomotaka.kuwahara@riken.jp}
\affiliation{
	Analytical Quantum Complexity RIKEN Hakubi Research Team, RIKEN Center for Quantum Computing (RQC), Wako, Saitama 351-0198, Japan
}
\affiliation{
	RIKEN Cluster for Pioneering Research (CPR), Wako, Saitama 351-0198, Japan 
}
\affiliation{
	PRESTO, Japan Science and Technology (JST), Kawaguchi, Saitama 332-0012, Japan}

\author{Keiji Saito}
\email{keiji.saitoh@scphys.kyoto-u.ac.jp}
\affiliation{Department of Physics, Kyoto University, Kyoto 606-8502, Japan}

\begin{abstract}
The area law of the bipartite information measure characterizes one of the most fundamental aspects of quantum many-body physics. 
In thermal equilibrium, the area law for the mutual information universally holds at arbitrary temperatures as long as the systems have short-range interactions.
In systems with power-law decaying interactions, $r^{-\alpha}$ ($r$: distance), conditions for the thermal area law are elusive.
In this work, we aim to clarify the optimal condition $\alpha> \alpha_c$ such that the thermal area law universally holds. 
A standard approach to considering the conditions is to focus on the magnitude of the boundary interaction between two subsystems.
However, we find here that the thermal area law is more robust than this conventional argument suggests.  
We show the optimal threshold for the thermal area law by $\alpha_c= (D+1)/2$ ($D$: the spatial dimension of the lattice), assuming a power-law decay of the clustering for the bipartite correlations. Remarkably, this condition encompasses even the thermodynamically unstable regimes $\alpha < D$.
We verify this condition numerically, finding that it is qualitatively accurate for both integrable and nonintegrable systems. 
Unconditional proof of the thermal area law is possible by developing the power-law clustering theorem for $\alpha > D$ above a threshold temperature.
Furthermore, the numerical calculation for the logarithmic negativity shows that the same criterion $\alpha > (D+1)/2$ applies to the thermal area law for quantum entanglement.
\end{abstract}

\maketitle

\textit{Introduction.---}
Quantum correlation and entanglement play pivotal roles in understanding quantum many-body systems from an information-theoretic standpoint. 
They help in identifying the exotic quantum phases~\cite{PhysRevLett.96.110404,PhysRevLett.96.110405,PhysRevB.82.155138} and serve as crucial resources in quantum information processing~\cite{RevModPhys.81.865}. One hallmark of quantum many-body systems is the area law for the ground state~\cite{RevModPhys.82.277}, which posits that the correlation and entanglement between two subsystems are constrained by the surface area of their interface. This principle has undergone extensive verification through both analytical and numerical approaches in a wide array of scenarios~\cite{PhysRevLett.90.227902,PhysRevA.66.042327,PhysRevLett.94.060503,PhysRevLett.96.100503,Hastings_2007,PhysRevA.80.052104,PhysRevB.85.195145,arad2013area,brandao2013area,kuwahara2019area,10.1145/3519935.3519962}. The area law not only elucidates the inherent complexities of quantum systems but also significantly enhances their simulatability using classical computers~\cite{landau2015polynomial,Arad2017}.

The quantum Gibbs state describing thermal equilibrium obeys an analogous principle, known as the thermal area law, as long as the systems have short-range interactions~\cite{PhysRevLett.100.070502, PhysRevE.93.022128, PhysRevB.98.235154, PhysRevX.11.011047,Lemm2023thermalarealaw}. 
The thermal area law for the mutual information is written as
\begin{align}
	\mathcal{I}(A:B) = \orderof{\beta |\partial A|} , \label{wolf}
\end{align}
where $\mathcal{I}$ represents the mutual information of the Gibbs state between two subsystems $A$ and $B$ at inverse temperature $\beta$.
Here, $|\partial A|$ denotes the size of the boundary between $A$ and $B$.
The thermal area law is not only relevant to the mutual information but also extends to a wider range of bipartite information measures~\cite{PhysRevE.93.022128,PhysRevB.98.235154,PhysRevX.11.011047}. We note that a recent study has further clarified that genuine quantum correlations for the Gibbs state also cannot be as large as the system size for short-range interacting systems~\cite{kuwahara2022exponential}.

\begin{figure}[tt]
	\centering
	\includegraphics[width=0.45\textwidth]{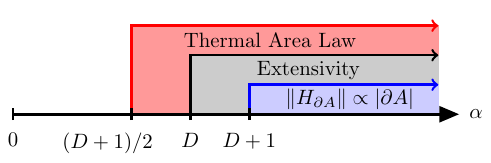}
	\caption{Schematic of the results. The thermal area law is satisfied for $\alpha > \alpha_c=(D+1)/2$ (indicated by a red region). 
		The thermal area law is more robust than expected, with the argument involving the boundary interaction $H_{\partial A}$ (blue), which represents the interaction between $A$ and $B$. The operator norm $\|H_{\partial A}\|$ scales with $|\partial A|$ for $\alpha > D+1$.
		We emphasize that the key physics behind finding the threshold $\alpha_c$ is the power-law clustering of bipartite correlations.
	}
	\label{fig:outline}
\end{figure}

In contrast to short-range interactions, the general aspects of locality in long-range interacting systems remain less understood. Long-range interactions here refer to interactions that decay as $r^{-\alpha}$ with distance $r$ \cite{RevModPhys.95.035002}, which are ubiquitous in nature such as in atomic, molecular, and optical systems \cite{yan2013observation,jurcevic2014quasiparticle,richerme2014non,Islam583,bernien2017probing,zhang2017observation}.
These interactions lead to nontrivial phenomena absent in short-range systems \cite{PhysRevLett.109.267203,PhysRevLett.113.156402,PhysRevB.93.205115,PhysRevLett.123.210602,PhysRevX.10.031010,tran2021optimal,tran2021liebrobinson,PhysRevLett.128.010603,tamaki2020energy}.
Key unresolved issues include (i) identifying the critical threshold $\alpha_{\rm c}$ above which 
any systems have the thermal area law, (ii) understanding the underlying physics, and (iii) distinguishing between classical and quantum correlations.
This paper addresses these questions as fundamental physics.
In a practical aspect, this direction may lead to finding an efficiency-guaranteed algorithm for simulating the long-range quantum Gibbs state.
The mutual information is bounded by the interaction Hamiltonian as shown in Ref.~\cite{PhysRevLett.100.070502}. A simple argument on the operator norm \cite{kuwahara2019area,PhysRevLett.128.010603} leads to that the thermal area law~\eqref{wolf} still remains valid for the condition $\alpha>D+1$. At first glance, this may lead to an intuition that the optimal condition for the thermal area law is given by $\alpha>D+1$ (or $\alpha_c=D+1$).
Indeed, the threshold given by the finite boundary interaction recurs across various contexts in physics, indicating that physical properties undergo qualitative changes at this threshold in various many-body systems ~\cite{Dyson1969,PhysRevLett.37.1577,Araki1975,Kishimoto1976,PhysRevLett.119.050501,PhysRevLett.128.010603}.  

However, in this work, we discover that the underlying physics of the thermal area law is influenced by a different characteristic: {\it the clustering property}, which is expected to remain robust in any noncritical phases. 
Note that the clustering property has been widely adopted as a basic and natural assumption in discussing various quantum properties such as entanglement area law~\cite{brandao2013area}, quantum Gibbs sampling~\cite{Kastoryano2016,Brandao2019}, the eigenstate thermalization~\cite{brandao2015equivalence,PhysRevLett.124.200604}, etc.
Consequently, we show that the thermal area law for the mutual information possesses a greater resilience than traditionally anticipated. Utilizing the power-law clustering properties in bipartite correlations, we analytically find that the sufficient condition for the thermal area law in the mutual information is given by $\alpha > (D + 1)/2$. Regarding the critical threshold $\alpha_{c}$ above which any systems satisfy the thermal area law, our numerical analysis concludes that $\alpha_{c} = (D+1) / 2$, since below $\alpha_c$ we provide explicit examples violating the area law.
Remarkably, this condition encompasses $\alpha<D$, where the thermodynamic extensivity is violated~\cite{RevModPhys.95.035002}. In this regime, most analytical techniques break down, and universal behaviors remain largely unexplored. 
While the numerical investigation suggests that the clustering property robustly holds even for thermodynamic nonextensive regimes, we demonstrate that the power-law clustering theorem can be rigorously derived for $\alpha > D$ by identifying a suitable temperature regime. Consequently, the thermal area law unconditionally holds for $\alpha > D$. Our key results are depicted schematically in Fig.~\ref{fig:outline}.

In addition, we investigate whether differences exist between quantum and classical correlations in the context of the thermal area law. By employing the logarithmic negativity proposed by Shapourian, Shiozaki, and Ryu for bilinear fermionic systems~\cite{PhysRevB.95.165101},  we demonstrate that this criterion equally governs the thermal area law for quantum entanglement.

\textit{Setup.---}
We consider a quantum system with $N$ qudits or fermions on the $D$-dimensional lattice. Let $\Lambda$ be the set of all sites on the lattice of $N$ sites. We consider the Hamiltonian $H$ in the $k$-local form:
\begin{align}
	H = \sum_{Z:|Z|\le k} h_Z \, , \label{ham}
\end{align}
where $h_{Z}$ is the local Hamiltonian on a subset $Z \subset \Lambda$ with site number $|Z|$.
We consider the long-range interactions satisfying the following condition:
\begin{align}
	J_{i,i'}:= \sum_{Z:Z\ni \{i,i'\}}\norm{ h_Z} \le  
	{g\over (1+\dist_{i,i'})^{\alpha} } \, . \label{main_interaction_cond}
\end{align}
Here, $J_{i,i'}$ is the maximum norm of the local Hamiltonian including the sites $i$ and $i'$, and $g$ is some constant. The symbol $d_{i,i'}$ stands for the distance between the sites $i$ and $i'$. The parameter $\alpha$ is the index of power-law decay in the interaction. We will provide examples of the Hamiltonian in Eqs.~(\ref{Model_1}) and (\ref{Model_2}) later. Note that the Hamiltonian does not contain the Kac factor for the thermodynamically nonextensive regime $\alpha <D$, since such a factor does not exist in realistic situations~\cite{richerme2014non,neyenhuis2017observation}.

The quantum Gibbs state $\rho_{\beta}$ at inverse temperature $\beta$ is defined as $\rho_{\beta} := e^{-\beta H} / \tr(e^{-\beta H})$. The mutual information of $\rho_{\beta}$ between two regimes $A$ and $B$ is defined as 
\begin{align}
	\mathcal{I}_{\rho_{\beta}}(A:B) := S(\rho_{\beta}^{A}) + S(\rho_{\beta}^{B}) - S(\rho_{\beta}^{AB}),
\end{align}
where $\rho_{\beta}^{X}$ is the reduced density matrix of $\rho_{\beta}$ on the subset $X\subseteq \Lambda$, and $S(\rho_{\beta}^{X}) := - \tr(\rho_{\beta}^{X} \log \rho_{\beta}^{X})$ is the von Neumann entropy. 
We define the correlation function of two observables $O_{A}$ and $O_{B}$ as 
\begin{align}
	\Cor_{\rho_{\beta}}(O_{A},O_{B}) := \tr(\rho_{\beta} O_{A} O_{B}) - \tr(\rho_{\beta} O_{A}) \cdot \tr(\rho_{\beta} O_{B}).
\end{align}

\textit{Main theorem.---}
We find that the thermal area law of the mutual information holds above some power-law threshold with the assumption of clustering. 

\textbf{Theorem 1:} {\it Let us assume that the correlation function of the quantum Gibbs state $\rho_{\beta}$ between two arbitrary operators $O_{i}$ and $O_{i'}$ supported on the sites $i$ and $i'$ satisfies the following power-law clustering property:
	\begin{align}
		\left|\Cor_{\rho_{\beta}}(O_{i},O_{i'})\right| \leq \frac{C}{d_{i,i'}^{\alpha}} \Vert O_{i} \Vert \cdot \Vert O_{i'} \Vert, \label{clustering}
	\end{align}
	where $C$ is an $\orderof{1}$ constant. Then, for the bipartition $A,B$ of $D$-dimensional lattice $\Lambda$ $(A \cup B = \Lambda)$ and $\alpha > (D+1) / 2$, the mutual information is upper bounded by
	\begin{align}
		\mathcal{I}_{\rho_{\beta}}(A:B) \leq \text{const.} \beta | \partial A | \, . 
		\label{main_thm}
\end{align}}
We postpone the sketch of proof to the end of the paper (the details are in the Supplemental Material (SM)~\cite{Supplement_thermal}). The theorem ensures that the thermal area law holds for the regime $\alpha > (D+1)/2$. This condition means that the thermal area law is much more robust than expected from a simple argument on the norm of the boundary interaction, which leads to the weaker condition of $\alpha>D+1$.
Remarkably, the condition partially includes the thermodynamically nonextensive regime, extending up to interaction decays of $r^{-1}$ (1D), $r^{-3/2}$ (2D), and $r^{-2}$ (3D).

\textit{Power-law clustering and unconditional thermal area law.---}
Theorem 1 suggests that the crucial physics behind the thermal area law is not only from the magnitude of the boundary interaction but also from the clustering property. Note that the clustering property is expected to be one of the robust physical properties in noncritical thermal phases, rendering it an inherently natural assumption~\cite{brandao2013area}. We here demonstrate that the power-law clustering property can be rigorously proven for $\alpha > D$ above a temperature threshold as shown in the following Theorem 2.%We ext 

\textbf{Theorem 2:} {\it Under the general Hamiltonian (\ref{ham}) with (\ref{main_interaction_cond}) for the regime $\alpha > D$, the following clustering property holds for the temperatures above $\beta_c^{-1}$:
	\begin{align}
		\left|\Cor_{\rho_{\beta}}(O_{X},O_{Y})\right|  &\leq C \Vert O_{X} \Vert \cdot \Vert O_{Y} \Vert 
		{ |X| |Y| \, e^{{(|X| + |Y|)/ k}} \over d_{X,Y}^{\alpha}} \, , 
		\label{theor2}
	\end{align}
	where $d_{X,Y}=\min_{i\in X , j\in Y} d_{i,j}$ and $|X|$ and $|Y|$ are numbers of sites in the regions $X$ and $Y$, respectively. $C$ is an $\orderof{1}$ constant and the threshold temperature is given by $\beta_c =1/(8ugk)$ with $u=\sum_{j\in \Lambda} 1/(1+d_{i,j})^{\alpha}$.}

The proof is based on the cluster expansion technique. The details in the proof are provided in SM \cite{Supplement_thermal}. To the best of our knowledge, this is the first result to establish the power-law clustering theorem in quantum long-range interacting systems.

From this statement, the clustering property in Theorem 1 is not an assumption but rigorously holds above a threshold temperature as long as the thermodynamically extensive regime (i.e., $\alpha>D$) is considered. In particular, in 1D systems, the clustering property is believed to hold at arbitrary temperatures~\cite{Araki1969,kimura2024clustering}. Moreover, we stress that even in the thermodynamically nonextensive regime ($\alpha <D$), numerical calculations for the two models below show the validity of the clustering property.

\textit{Numerical verification.---}
We numerically verify our theorems in both integrable and nonintegrable systems, establishing the bound's tightness.

As a typical example of integrable systems, we use the following long-range bilinear fermion system:
\begin{align}
	H = - \sum_{i,j \in \Lambda} \frac{t_{i,j}}{\dist_{i,j}^{\alpha}} (c_{i}^{\dagger} c_{j} + c_{j}^{\dagger} c_{i}), \label{Model_1}
\end{align}
where $t_{i,j}$ is a hopping parameter of order $\mathcal{O}(1)$, and $c_{i}$ and $c_{i}^{\dagger}$ are the annihilation and creation operators of the spinless fermion at site $i$, respectively. We consider the Hamiltonian in one dimension (1D) and two dimensions (2D).

We first verify the clustering property~\eqref{clustering} in the Hamiltonian~\eqref{Model_1}~\cite{Latorre_2009,PhysRevB.94.214301}.
We analyze a 1D chain of $N = 1000$ sites and a 2D square lattice of side length $N = 40$, focusing on two-point correlation functions $\langle c_{i}^{\dagger} c_{i+r} \rangle$ in 1D and $\langle c_{{\bm i}}^{\dagger} c_{{\bm i}+{\bm r}} \rangle$ in 2D at $\beta = 2$. 
For general bilinear systems, we assign random values $\in [0,1]$ to the hopping variable and average over 1000 samples.
In Figs.~\ref{fig:bilinear_clustering_MI}(a) and (b), $|\langle c_{i}^{\dagger} c_{i+r} \rangle| \times r^{\alpha}$ (1D) and $|\langle c_{{\bm i}}^{\dagger} c_{{\bm i}+{\bm r}} \rangle| \times |{\bm r}|^{\alpha}$ (2D) saturate at large distances for any $\alpha$, confirming clustering in both dimensions, even for $\alpha <D$.
While we show the averaged data, we checked that individual data also satisfy the clustering property (not shown). 
All correlation functions in bilinear systems follow this property via the Wick theorem.

After establishing the clustering property, we affirm the thermal area law applies for $\alpha > (D+1)/2$, a sufficient condition whose optimality requires further investigation. 
We assess its tightness through numerical calculations on the same model.
The system is partitioned into subsystems $A$ and $B$, with $A$ as the first $N/2$ sites and $B$ the rest in 1D. For 2D, $A$ includes the first $(N \times N/2)$ sites, and $B$ the remainder. 
The mutual information $\mathcal{I}_{\rho_{\beta}}(A:B)$ is calculated with the same parameters as in the clustering property study.
Results are shown in Figs.~\ref{fig:bilinear_clustering_MI}(c) for 1D and~\ref{fig:bilinear_clustering_MI}(d) for 2D.
The mutual information increases with the system size if $\alpha < 1$ and it saturates to a constant value if $\alpha > 1$ in 1D.
The increasing trend of the mutual information, fitted to a power law, is detailed in SM (see Fig. S.2)~\cite{Supplement_thermal}.
In 2D, the mutual information divided by the boundary area $N$ grows for $\alpha < 1.5$ and steadies for $\alpha > 1.5$. 
These numerical results support that the thermal area law holds for $\alpha > (D+1) / 2$, and hence the condition is optimal.

\begin{figure}[t]
	\centering
	\includegraphics[width=0.5\textwidth]{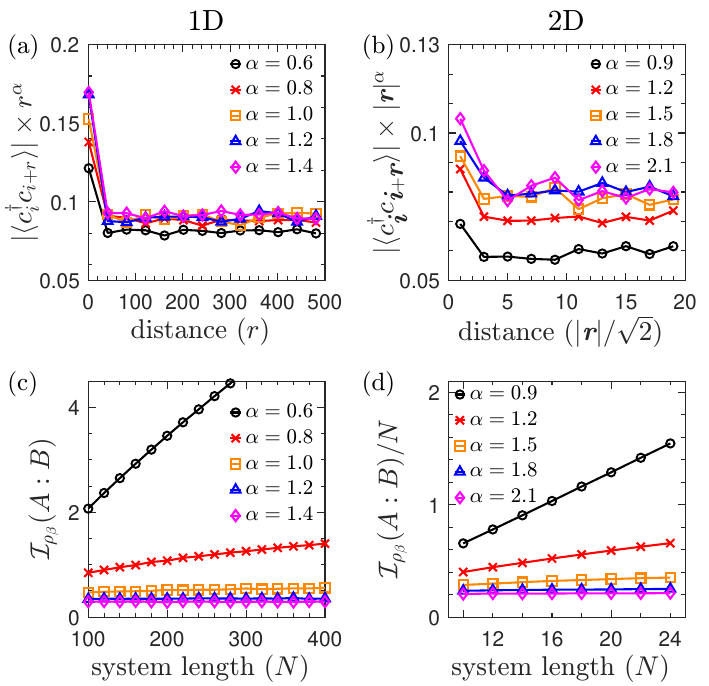}
	\caption{1D and 2D long-range bilinear fermions. (a) The two-point correlation function $|\langle c_{i}^{\dagger} c_{i+r} \rangle|$ in the 1D chain of $N = 1000$ sites, between the sites $i = 250$ and $i + r$, multiplied by $r^{\alpha}$ for various $\alpha$.
		(b) The two-point correlation function $|\langle c_{{\bm i}}^{\dagger} c_{{\bm i}+{\bm r}} \rangle|$ in the 2D square lattice of side length $N = 40$, between the sites ${\bm i} = (10,10)$ and ${\bm i} + {\bm r} = (10 + r,10 + r)$, multiplied by $|{\bm r}|^{\alpha}$ for various $\alpha$.
		(c) Mutual information $\mathcal{I}_{\rho_{\beta}}(A:B)$ in the 1D chain of $N$ sites between halves $A$ and $B$.
		(d) Mutual information $\mathcal{I}_{\rho_{\beta}}(A:B)$ in the 2D square lattice of side length $N$, between halves $A$ and $B$, divided by $N$.
		All figures are the average over 1000 samples of random variables $t_{i,j} \in [0,1]$ in Eq.~\eqref{Model_1} at $\beta = 2$.}
	\label{fig:bilinear_clustering_MI}
\end{figure}

In general, bilinear systems can be special in several physical aspects not only for the nonequilibrium properties~\cite{d2016quantum} but also for the static properties including the thermal entanglement~\cite{PhysRevA.78.022103,Bernigau_2015,barthel2017one}.
However, the clustering is expected to be robust, regardless of the (non)integrability, and hence the thermal area law should also hold universally.
We check this universality in a specific nonintegrable system: the 1D long-range Heisenberg spin-1/2 chain with $N$ sites,
\begin{align}
	H = \sum_{1 \leq i<j \leq N} \frac{1}{\dist_{i,j}^{\alpha}} {\bm S}_{i} \cdot {\bm S}_{j}, \label{Model_2}
\end{align}
where ${\bm S}_{i} = (S_{i}^{x},S_{i}^{y},S_{i}^{z})$ is the spin-1/2 operator  at site $i$.
We explore the clustering property of $\langle {\bm S}_{i} \cdot {\bm S}_{i+r} \rangle$ for $N = 200$ and $\beta = 2$ using the exponential tensor renormalization group (XTRG) algorithm~\cite{PhysRevX.8.031082,PhysRevB.100.045110} and the QSpace tensor library~\cite{Weichselbaum2012,PhysRevResearch.2.023385} to construct the matrix product operator for the Gibbs state $\rho_{\beta}$.
See Ref.~\cite{ft0} for details.
Figure~\ref{fig:Heisenberg_clustering_TFDEE}(a) shows $|\langle {\bm S}_{i} \cdot {\bm S}_{i+r} \rangle| \times r^{\alpha}$, demonstrating clustering via constant upper bounds.
Furthermore, we confirm the clustering property in the two-dimensional long-range Heisenberg and XX models (see SM~\cite{Supplement_thermal}). 
These findings strongly suggest that clustering universally holds.

\begin{figure}[t]
	\centering
	\includegraphics[width=0.5\textwidth]{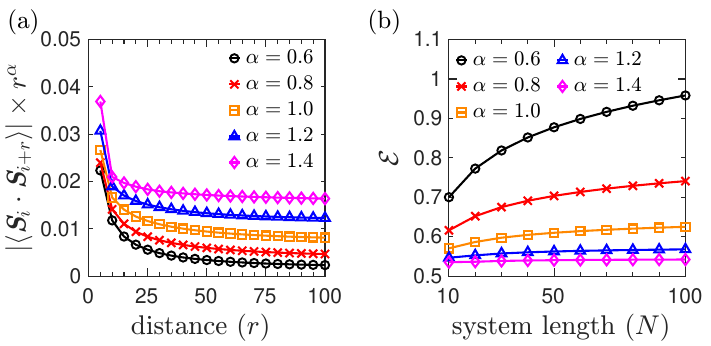}
	\caption{1D long-range Heisenberg chain. (a) The two-point correlation function $|\langle {\bm S}_{i} \cdot {\bm S}_{i+r} \rangle|$ multiplied by $r^{\alpha}$ in the 1D chain with $N = 200$ sites, between the sites $i = 50$ and $i+r$, at $\beta = 2$ for various $\alpha$.
		(b) The entanglement entropy Eq.~\eqref{TFDEE} of the TFD state of the Gibbs state at $\beta = 2$ is computed between the two halves of the system for various $\alpha$ with varying $N$.
		We employ the XTRG algorithm with the numerical details in Ref.~\cite{ft0}.}
	\label{fig:Heisenberg_clustering_TFDEE}
\end{figure}

To verify the thermal area law, we compute the entanglement entropy ($\mathcal{E}$) for the thermofield double (TFD) state of Gibbs state $\rho_{\beta}$ as it upper bounds the mutual information~\cite{Cottrell2019}.
We consider two identical copies $\mathcal{H}_{L}$ and $\mathcal{H}_{R}$ of the original Hilbert space $\mathcal{H}$.
For the eigenvalues $\{E_{n}\}$ and eigenstates $\{|n \rangle\}$ of the Hamiltonian, the TFD state of the Gibbs state $\rho_{\beta} = \sum_{n} e^{-\beta E_{n}} / \tr(e^{-\beta H}) |n \rangle \langle n|$ is defined as 
\begin{align}
	|\mathrm{TFD} \rangle := \frac{1}{\sqrt{\tr(e^{-\beta H})}} \sum_{n} e^{-\beta E_{n} / 2} |n \rangle_{L} \otimes |n \rangle_{R}.
\end{align}
Partitioning the system into halves $A$ and $B$, we define $\mathcal{E}$ as the von Neumann entropy of the reduced density matrix $\sigma_{A_{L},A_{R}} = \tr_{B_{L},B_{R}}(|\mathrm{TFD} \rangle \langle \mathrm{TFD}|)$:
\begin{align}
	\mathcal{E} := S(\sigma_{A_{L},A_{R}}) = - \tr(\sigma_{A_{L},A_{R}} \log \sigma_{A_{L},A_{R}}), \label{TFDEE}
\end{align}
where $A_{L}$ and $A_{R}$ ($B_{L}$ and $B_{R}$) are the copies of the original subsystem $A$ ($B$). 
Then the mutual information is upper-bounded by $2 \mathcal{E}$, i.e. $\mathcal{I}_{\rho_{\beta}}(A:B) \leq 2 \mathcal{E}$~\cite{PhysRevA.91.042323,PhysRevX.11.011047}.
We show the results of $\mathcal{E}$ in Fig.~\ref{fig:Heisenberg_clustering_TFDEE}(b). 
$\mathcal{E}$ trends upward for $\alpha < 1$ and stabilizes for $\alpha > 1$, supporting the thermal area law's optimality condition.
While this discussion focuses solely on Hamiltonian~\eqref{Model_2}, it is worth noting that similar behaviors are numerically demonstrated even when the disorder is added into the Hamiltonian (see SM~\cite{Supplement_thermal}).

\textit{Thermal area law in the quantum entanglement.---}
The mutual information includes both classical and quantum correlations. 
We now extract purely quantum correlations from the Gibbs state to analyze the thermal area law condition.
To this end, we investigate bilinear fermionic systems~\eqref{Model_1} using Shapourian, Shiozaki, and Ryu's (SSR) logarithmic negativity~\cite{PhysRevB.95.165101} to quantify quantum entanglement between subsystems $A$ and $B$ in a mixed state efficiently. 
The SSR logarithmic negativity is formulated through the partial time-reversal transform $R_{A}$, yielding
\begin{align}
	E_{\mathrm{SSR}}(\rho) := \log \Vert \rho^{R_{A}} \Vert_{1} \, ,
\end{align}
where $\Vert \cdots \Vert_{1}$ is the trace norm.
Details on the SSR logarithmic negativity for bilinear systems are in SM~\cite{Supplement_thermal}. We calculate this for 1D and 2D Hamiltonian systems~\eqref{Model_1} with identical parameters and partitioning as the mutual information, showing results in Figs.~\ref{fig:SSR_Negativity}(a) for 1D and~\ref{fig:SSR_Negativity}(b) for 2D.
These figures demonstrate that the SSR logarithmic negativity shows the same behavior as the mutual information in Figs~\ref{fig:bilinear_clustering_MI}(c) and~\ref{fig:bilinear_clustering_MI}(d).
Hence, the figures indicate that under the same conditions, namely $\alpha > (D+1) / 2$, quantum entanglement likewise follows the thermal area law.

\begin{figure}[t]
	\centering
	\includegraphics[width=0.5\textwidth]{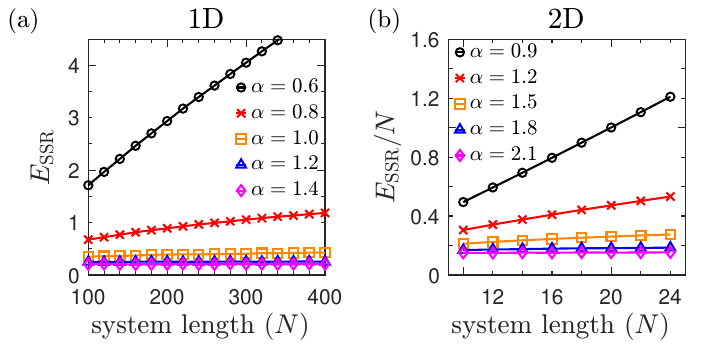}
	\caption{The SSR negativity $E_{\mathrm{SSR}}$ of long-range bilinear fermions. (a) $E_{\mathrm{SSR}}$ of the 1D chain of $N$ sites and (b) $E_{\mathrm{SSR}} / N$ of the 2D square lattice with a side length $N$ between one half of the system and the other half for various $\alpha$.
		The figures are averaged over the 1000 samples of random variables $t_{i,j} \in [0,1]$ in Eq.~\eqref{Model_1} at $\beta = 2$.}
	\label{fig:SSR_Negativity}
\end{figure}

\textit{Proof sketch of main theorem.---}
We here provide the outline of the proof of the main theorem. 
For the sake of simplicity, we focus on the simple case where $h_{Z}$ is $h_{i,j}$ that acts only on two different sites $i,j \in \Lambda$. The proof for the general case is in SM \cite{Supplement_thermal}.

The first step is to leverage the remarkable inequality demonstrated in Ref.~\cite{PhysRevLett.100.070502}.
Consider the total Hamiltonian $H = H_{A} + H_{B} + H_{\partial A}$, where $H_{A}$, $H_{B}$ are supported only on regions $A$ and $B$ respectively, and $H_{\partial A}$ represents the interaction  between these regions.
Reference~\cite{PhysRevLett.100.070502} utilized the Gibbs variational principle, expressed as
\begin{align}
	F(\rho ) - F(\rho_{\beta}) & = \beta^{-1} D(\rho || \rho_{\beta} ) \ge 0 \, , 
\end{align}
where $F(\rho):= \tr(H \rho )+ \beta^{-1}\tr [\rho \log (\rho)]$ is the nonequilibrium free energy and $D(\rho || \sigma)=\tr (\rho \log \rho) - \tr (\rho \log \sigma)$ is the quantum relative entropy. 
By substituting $\rho = \rho_{\beta}^{A} \otimes \rho_{\beta}^{B}$ into this inequality, Ref.~\cite{PhysRevLett.100.070502} established a bound on the mutual information:
\begin{align}
	\mathcal{I}_{\rho_{\beta}}(A:B) &\leq \beta \, \mathrm{tr}\left[(\rho_{\beta}^{A} \otimes \rho_{\beta}^{B} - \rho_{\beta}^{AB}) H_{\partial A}\right]. \label{boundary_int_ineq}
\end{align}

The second crucial step is to reduce quantum complexity by utilizing the clustering property.
We can represent the boundary interaction as $H_{\partial A} = \sum_{i \in A} \sum_{j \in B} h_{i,j}$, and each operator $h_{i,j}$ can be expanded as $h_{i,j} = \sum_{s=1}^{d_{0}^{4}} h_{i}^{(s)} \otimes h_{j}^{(s)}$, with $h_{i}^{(s)}$ exclusively supported on site $i$ and $h_{j}^{(s)}$ on $j$.
Here, $d_0$ denotes the Hilbert space dimension of a single site. The operators $h_{i}^{(s)} \otimes h_{j}^{(s)}$ can be chosen to be Hermitian and mutually orthogonal with respect to the Hilbert-Schmidt inner product.
Inserting this decomposition into Eq.~\eqref{boundary_int_ineq} yields
\begin{align}
	\mathcal{I}_{\rho_{\beta}}(A:B) &\leq \beta \sum_{i \in A}\sum_{j \in B}\sum_{s = 1}^{d_{0}^{4}} \left|\mathrm{Cor}_{\rho_{\beta}}(h_{i}^{(s)},h_{j}^{(s)})\right| \, . 
\end{align}
The assumption of the clustering property (\ref{clustering}) leads to
\begin{align}
	\mathcal{I}_{\rho_{\beta}}(A:B) \leq C \sum_{i \in A} \sum_{j \in B}  \sum_{s = 1}^{d_{0}^{4}} \beta \Vert h_{i}^{(s)} \Vert\cdot \Vert h_{j}^{(s)} \Vert d_{i,j}^{-\alpha} . \label{derivation_1}
\end{align}
By the condition \eqref{main_interaction_cond}, we have
\begin{align}
	\max_{s} \br{\Vert h_{i}^{(s)} \Vert \cdot \Vert h_{j}^{(s)} \Vert} \leq d_{0}  J_{i,j}  \leq \frac{d_{0}g}{d_{i,j}^{\alpha}}. \label{derivation_2}
\end{align}
Combining Eq.~\eqref{derivation_1} and Eq.~\eqref{derivation_2}, we finally obtain
\begin{align}
	\mathcal{I}_{\rho_{\beta}}(A:B) \leq \sum_{i \in A} \sum_{j \in B} \frac{\beta d_{0}^{5} g C}{d_{i,j}^{2\alpha}}. \label{derivation_3} 
\end{align}
It can be proven that the summation of the distance $\sum_{i \in A} \sum_{j \in B} d_{i,j}^{- 2\alpha}$ is upper bounded by the boundary area when $2 \alpha > D+1$. Therefore, inequality~\eqref{derivation_3} leads to the desired inequality \eqref{main_thm}.

\textit{Summary and outlook.---}
We consider the validity of the thermal area law in the systems with long-range interactions $r^{-\alpha}$. Under the assumption of the clustering property, we derive the critical threshold $\alpha_{\rm c}:= (D+1)/2$ above which any systems obey the thermal area law of the mutual information (See Fig.\ref{fig:outline}). Remarkably, the regime covers the thermodynamically nonextensive regime.  
This criteria potentially may allow for an efficient representation of quantum Gibbs states exhibiting a power-law decay up to $r^{-\alpha_{\rm c}}$. 
Given this criterion, it is a crucial future problem  
to develop an efficiency-guaranteed algorithm for simulating long-range interacting systems at finite temperatures, as well as constructing tensor network states with polynomial bond dimensions~\cite{PhysRevB.91.045138,PhysRevX.4.031019}.
Furthermore, experimental verification of the thermal area law in long-range interacting systems appears feasible using the TFD state. A trapped ion quantum computer enables the preparation of the TFD state at finite temperatures, as outlined in Ref.~\cite{zhu2020generation}. Subsequent measurements of the Rényi entanglement entropy of this state, utilizing the methodology described in Ref.~\cite{islam2015measuring}, could potentially demonstrate its conformity to the thermal area law when $\alpha > (D+1)/2$.

\bibliography{LR_Thermal_Area}

\begin{thebibliography}{73}
\bibitem [{ft0()}]{ft0}%
  \BibitemOpen
  \href@noop {} {}We implement $\mathrm{SU}(2)$ spin symmetry of
  Eq.~\eqref{Model_2} with the tensor library QSpace. We set the bond
  dimensions of the matrix product operator (MPO) to 210 for
  Fig.~\ref{fig:Heisenberg_clustering_TFDEE}(a) and 120 for
  Fig.~\ref{fig:Heisenberg_clustering_TFDEE}(b). These bond dimensions are
  associated with the 210 and 120 $\mathrm{SU}(2)$ multiplets, which are
  roughly equivalent to $850$ and $450$ states.\BibitemShut {Stop}%
\end{thebibliography}

\section*{acknowledgments}

{~}\\
D. K. thanks M. L. Kim for helpful discussions about QSpace.
D. K. and T. K. acknowledge the Hakubi projects of RIKEN.  
T. K. was supported by JST PRESTO (Grant No. JPMJPR2116), ERATO (Grant No. JPMJER2302), and JSPS Grants-in-Aid for Scientific
Research (No. JP24H00071), Japan.
K. S. was supported by JSPS Grants-in-Aid for Scientific
Research (No. JP19H05603, No. JP19H05791 and No. JP23H01099).

\clearpage
\newpage

\renewcommand\thefootnote{*\arabic{footnote}}

%\counterwithout{section}{section}
\addtocounter{section}{0}

%\counterwithout{equation}{section}
\addtocounter{equation}{-19}

\renewcommand{\theequation}{S.\arabic{equation}}

\renewcommand{\thesection}{S.\Roman{section}}
\begin{widetext}

\begin{center}
	{\large \bf Supplemental Material for  ``Thermal Area Law in Long-Range Interacting Systems''}\\
	\vspace*{0.3cm}
	Donghoon Kim$^{1}$, Tomotaka Kuwahara$^{1,2,3}$, and Keiji Saito$^{4}$ \\
	\vspace*{0.1cm}
	$^{1}${\small \it Analytical Quantum Complexity RIKEN Hakubi Research Team, RIKEN Center for Quantum Computing (RQC), Wako, Saitama 351-0198, Japan} \\
	$^{2}${\small \it RIKEN Cluster for Pioneering Research (CPR), Wako, Saitama 351-0198, Japan}\\
	$^{3}${\small \it PRESTO, Japan Science and Technology (JST), Kawaguchi, Saitama 332-0012, Japan}\\
	$^{4}${\small \it Department of Physics, Kyoto University, Kyoto 606-8502, Japan}
\end{center}

\tableofcontents

\section{Setup}
\label{sec2}

We consider a quantum system, where we use $\Lambda$ to represent the set of sites. 
Given an arbitrary subset $X$ in $\Lambda$, $X\subseteq \Lambda$, we denote the cardinality of the set $X$, i.e., the number of sites contained in $X$, by $|X|$. 
For any two subsets $X, Y$ in $\Lambda$, $X, Y \subseteq \Lambda$, we define the distance between the two sets, $\dist_{X,Y}$, as the shortest path length of the graph that connects $X$ and $Y$. We remark that when $X$ intersects with $Y$, $X\cap Y \neq \emptyset$, the distance between the sets $X$ and $Y$ is zero: $\dist_{X,Y}=0$. 
%When $X$ consists of a single element, i.e., when $X=\{i\}$, we denote $\dist_{\{i\},Y}$ as $\dist_{i,Y}$ for brevity.
We denote the complementary set of $X$ as $X^\co$, namely $X^\co:= \Lambda\setminus X$, and the surface subset of $X$ as $\partial X$, i.e. $\partial X:=\{ i\in X| \dist_{i,X^\co}=1\}$, respectively. 
Also, for an arbitrary operator $O$, we denote the support of $O$ by $\Supp(O)$. 
We often describe the support of the operator explicitly by adding an Index of lower right subscript such as $O_X$ with $\Supp(O_X)=X$.

We consider a $k$-local Hamiltonian $H$ with long-range interactions: 
\begin{align}
	H = \sum_{Z:|Z|\le k} h_Z \, ,  \label{def:Hamiltonian}
\end{align}
where $Z$ is a subset indicating the interacting sites, and $h_Z$ is the local Hamiltonian. Note that $|Z|$ is the number of the interacting sites, and hence $h_Z$ means $Z$-body interaction. We assume that the interaction form satisfies the following property
\begin{align}
	\label{def_short_range_long_range}
	J_{i,i'}:= \sum_{Z:Z\ni \{i,i'\}}\norm{ h_Z} \le  
	{ g \over (1+\dist_{i,i'})^{\alpha} } .
\end{align} 
Here, $J_{i,i'}$ is the maximum operator norm of the local Hamiltonians containing the sites $i$ and $i'$. 

We consider the quantum Gibbs state $\rho_{\beta}$ with a fixed inverse temperature $\beta$ defined as follows:
\begin{align}
	\rho_{\beta} := e^{-\beta H}/Z ,\quad Z= \tr (e^{-\beta H}).
\end{align}
The mutual information of $\rho_{\beta}$ between two subsets $A$ and $B$ is defined as 
\begin{align}
	\mathcal{I}_{\rho_{\beta}}(A:B) := S(\rho_{\beta}^{A}) + S(\rho_{\beta}^{B}) - S(\rho_{\beta}^{AB}),
\end{align}
where $\rho_{\beta}^{X}$ is the density matrix of $\rho_{\beta}$ for the regime $X \in \{A,B,AB\}$ and $S(\rho_{\beta}^{X}) := - \tr(\rho_{\beta}^{X} \log \rho_{\beta}^{X})$ is the von Neumann entropy. 
The standard correlation function of the density matrix $\rho$ between two observables $O_{A}$ and $O_{B}$ is defined as
\begin{align}
	\Cor_{\rho}(O_{A},O_{B}) := \tr(\rho O_{A} O_{B}) - \tr(\rho O_{A}) \cdot \tr(\rho O_{B}).
\end{align}

\section{Proof of Main Theorem in a General Setting}

\begin{theorem}
	\label{main_thm_2}
	Let us assume that the correlation function of the quantum Gibbs state $\rho_{\beta}$ between two arbitrary operators $O_{X}$ and $O_{Y}$ supported on the subsets $X$ and $Y$ in $\Lambda$ satisfies the following power-law clustering property:
	\begin{align}
		\left|\Cor_{\rho_{\beta}}(O_{X},O_{Y})\right| \leq \frac{C}{d_{X,Y}^{\alpha}} \Vert O_{X} \Vert \cdot \Vert O_{Y} \Vert,
	\end{align}
	where $C$ is an $\orderof{1}$ constant.
	Then, for the bipartition $A,B$ of $D$-dimensional lattice $\Lambda$ ($A \cup B = \Lambda$) and $\alpha > (D+1) / 2$, the mutual information is upper bounded by
	\begin{align}
		\mathcal{I}_{\rho_{\beta}}(A:B) \leq \beta \cdot C \cdot |\partial A| \cdot \text{const}. \label{supp_main_thm}
	\end{align}
\end{theorem}
\begin{proof}
	We partition the lattice $\Lambda$ into two subsets $A,B \subset \Lambda$ and write the Hamiltonian as $H = H_{A} + H_{B} + H_{\partial A}$. Here, $H_{A}$, $H_{B}$ are the Hamiltonian supported only on $A,B$, respectively, and $H_{\partial A}$ is the interaction between $A$ and $B$, i.e., 
	\begin{align}
		H_{\partial A} = \sum_{\substack{Z : Z \cap A \neq \emptyset \\ Z \cap B \neq \emptyset}} h_{Z}. \label{supp_derivation_1}
	\end{align}
	Using the non-negativity of the quantum relative entropy $D(\rho \Vert \sigma) := \tr(\rho(\log \rho - \log \sigma))$ yields 
	\begin{align}
		D(\rho \Vert \rho_{\beta}) = \beta \tr(\rho H) - S(\rho) + \log Z = \beta F_{\beta}(\rho) - \beta F_{\beta}(\rho_{\beta}) \geq 0,
	\end{align}
	for an arbitrary density matrix $\rho$. Here, $F_{\beta}(\rho)$ is the free energy defined as
	\begin{align}
		F_{\beta}(\rho) := \tr(\rho H) - \beta^{-1}S(\rho),
	\end{align}
	and it becomes $F_{\beta}(\rho_{\beta}) = -\beta^{-1} \log Z$ at equilibrium. Substituting $\rho = \rho_{\beta}^{A} \otimes \rho_{\beta}^{B}$ and using $\mathcal{I}_{\rho_{\beta}}(A:B) = S(\rho_{\beta}^{A} \otimes \rho_{\beta}^{B}) - S(\rho_{\beta}^{AB})$ and $\tr((H_{A} + H_{B}) \rho_{\beta}^{A} \otimes \rho_{\beta}^{B}) = \tr((H_{A} + H_{B}) \rho_{\beta}^{AB})$, we obtain
	\begin{align}
		\mathcal{I}_{\rho_{\beta}}(A:B) \leq \beta \cdot \tr((\rho_{A} \otimes \rho_{B} - \rho_{AB}) H_{\partial A}). \label{supp_derivation_2}
	\end{align}
	For the $k$-local Hamiltonian, each local term $h_{Z}$ can be expanded by a sum of the tensor products of two operators supported on $Z_{A} = Z \cap A$ and $Z_{B} = Z \cap B$: 
	\begin{align}
		h_{Z} = \sum_{s = 1}^{d_0^{2k}} h_{Z_{A}}^{(s)} \otimes h_{Z_{B}}^{(s)}. \label{supp_derivation_3}
	\end{align}
	Here, $d_0$ denotes the Hilbert space dimension of a single site. The operators $h_{i}^{(s)} \otimes h_{j}^{(s)}$ can be chosen to be Hermitian and orthogonal to each other with respect to the Hilbert-Schmidt inner product, i.e. $\tr\br{(h_{Z_{A}}^{(s)} \otimes h_{Z_{B}}^{(s)})^{\dagger} (h_{Z_{A}}^{(s')} \otimes h_{Z_{B}}^{(s')})} \propto \delta_{s,s'}$.
	Equations~\eqref{supp_derivation_1},~\eqref{supp_derivation_2} and~\eqref{supp_derivation_3} lead to
	\begin{align}
		\mathcal{I}_{\rho_{\beta}}(A:B) \leq \sum_{\substack{Z : Z \cap A \neq \emptyset \\ Z \cap B \neq \emptyset}} \sum_{s=1}^{d_0^{2k}} \beta \cdot \tr\br{(\rho_{A} \otimes \rho_{B} - \rho_{AB}) h_{Z_{A}}^{(s)} \otimes h_{Z_{B}}^{(s)}} \leq \sum_{\substack{Z : Z \cap A \neq \emptyset \\ Z \cap B \neq \emptyset}} \sum_{s=1}^{d_0^{2k}} \beta \cdot \left|\mathrm{Cor}_{\rho_{AB}}(h_{Z_{A}}^{(s)},h_{Z_{B}}^{(s)})\right|.
	\end{align}
	The assumption of the clustering property gives
	\begin{align}
		\mathcal{I}_{\rho_{\beta}}(A:B) &\leq \sum_{\substack{Z : Z \cap A \neq \emptyset \\ Z \cap B \neq \emptyset}} \sum_{s=1}^{d_0^{2k}} \beta \cdot \Vert h_{Z_{A}}^{(s)} \Vert \Vert h_{Z_{B}}^{(s)} \Vert \frac{C}{\dist_{Z_{A},Z_{B}}^{\alpha}} \\
		&\leq \sum_{i \in A} \sum_{j \in B} \sum_{\substack{Z : Z_{A} \ni i \\ Z_{B} \ni j}}  \sum_{s=1}^{d_0^{2k}} \beta \cdot \Vert h_{Z_{A}}^{(s)} \Vert \Vert h_{Z_{B}}^{(s)} \Vert \frac{C}{\dist_{i,j}^{\alpha}} \\
		&\leq \sum_{i \in A} \sum_{j \in B} \sum_{\substack{Z : Z_{A} \ni i \\ Z_{B} \ni j}} \beta \cdot d_0^{2k} \max_{s} \Vert h_{Z_{A}}^{(s)} \Vert \Vert h_{Z_{B}}^{(s)} \Vert \frac{C}{\dist_{i,j}^{\alpha}}. \label{supp_derivation_4}
	\end{align}
	In the second inequality, we have used the fact that $\dist_{Z_{A},Z_{B}}^{-\alpha}$ can be expressed as $\dist_{i,j}^{-\alpha}$ for some $i \in Z_{A}$ and $j \in Z_{B}$. In the third inequality, we have used $\sum_{s = 1}^{d_0^{2k}} \Vert h_{Z_{A}}^{(s)} \Vert \Vert h_{Z_{B}}^{(s)} \Vert \leq d_0^{2k} \max_{s} \Vert h_{Z_{A}}^{(s)} \Vert \Vert h_{Z_{B}}^{(s)} \Vert$.
	Note that
	\begin{align}
		\max_{s} \Vert h_{Z_{A}}^{(s)} \Vert \Vert h_{Z_{B}}^{(s)} \Vert \leq \sqrt{\tr(h_{Z}^{\dagger} h_{Z})} \leq d_{0}^{k/2} \Vert h_{Z} \Vert.
	\end{align}
	By the condition Eq.~\eqref{def_short_range_long_range} of the long-range interaction, we obtain
	\begin{align}
		\sum_{\substack{Z : Z_{A} \ni i \\ Z_{B} \ni j}} \max_{s} \Vert h_{Z_{A}}^{(s)} \Vert \Vert h_{Z_{B}}^{(s)} \Vert \leq d_{0}^{k/2} \sum_{Z : Z \ni \{i,j\}} \Vert h_{Z} \Vert \leq \frac{d_{0}^{k/2} g}{\dist_{i,j}^{\alpha}}.\label{supp_derivation_5}
	\end{align}
	Combining Eq.~\eqref{supp_derivation_4} and Eq.~\eqref{supp_derivation_5},
	\begin{align}
		\mathcal{I}_{\rho_{\beta}}(A:B) \leq \sum_{i \in A} \sum_{j \in B} \frac{\beta d_0^{5k/2} g C}{\dist_{i,j}^{2\alpha}}. \label{supp_derivation_6}
	\end{align}
	It was proven that $\sum_{i \in A} \sum_{j \in B} \dist_{i,j}^{-2 \alpha}$ is upper bounded by the boundary area $|\partial A|$ when $2 \alpha > D+1$~\cite{PhysRevLett.128.010603}. With Eq.~\eqref{supp_derivation_6}, we finally obtain 
	\begin{align}
		\mathcal{I}_{\rho_{\beta}}(A:B) \leq \beta \cdot C \cdot |\partial A| \cdot \text{const},
	\end{align}
	for $2 \alpha > D + 1$.
\end{proof}

%%%%%

\section{Power-law clustering theorem for $\alpha >D $}
\label{sec3}
We consider the regime $\alpha > D$, so that the extensivity property holds as
\begin{align} 
	\label{extensive_cond}
	\widehat{H_L}&:= \sum_{Z: Z\cap L\neq \emptyset} h_Z = H- H_{L^\co} \, , \\
	\norm{\widehat{H_L}} &\le \sum_{i\in L} \sum_{Z \in i} \| h_Z \| \le gu |L| \, ,  
\end{align}
for an arbitrary subset $L\subseteq \Lambda$. Here $u:=2^{\alpha} \sum_{j\in \Lambda} 1/(1+ d_{i,j})^{\alpha}$.
%Given operators $O_X$ and $O_Y$ supported on the subsets $X$ and $Y$ in $\Lambda$, $X,Y\subseteq \Lambda$, the correlation function, ${\rm Cor}_{\rho}(O_X, O_Y)$, is defined as follows: 
%\begin{align}
%\Cor_\rho(O_X,O_Y) := \tr (\rho O_X O_Y) - \tr (\rho O_X)  \cdot  \tr (\rho O_Y). 
%\end{align}
Then, the clustering theorem for the bipartite correlations is rigorously proven as

\begin{theorem}
	{\it
		\label{main_thm_1}
		The following relation holds for long-range interacting systems for the temperatures above $\beta_c^{-1}$:
		\begin{align}
			\left|\Cor_{\rho_{\beta}}(O_{X},O_{Y})\right| &\leq C \Vert O_{X} \Vert \cdot \Vert O_{Y} \Vert 
			{ |X| |Y| \, e^{{(|X| + |Y|)/ k}} \over d_{X,Y}^{\alpha}} \, , 
			\label{theor2}
		\end{align}
		where $C$ is an $\orderof{1}$ constant and the threshold temperature is given by $\beta_c=W(1/2 e^2 )/ (2 g u k)$ with the Lambert $W$ function (i.e., the inverse function of $x e^x$) and $u=2^{\alpha} \sum_{j\in \Lambda} 1/(1+d_{i,j})^{\alpha}$. We can lower-bound $\beta_c$ as $\beta_c \ge 1/ (16  g ue k)$ with the relation $W(1/2e^2) > 1/(8e)$.
	}
\end{theorem}

\vspace{.2in}

%\section{Proof of clustering theorem for $\alpha >D$}
%The above condition is satisfied for 
%\begin{align}
%\alpha >  \begin{cases}  
	%D &\for \textrm{generic Hamiltonians} \\
	%D/2&\for \textrm{bilinear Hamiltonians}
	%\end{cases}
	%\end{align}

For the proof of the clustering theorem, we utilize the cluster expansion~\cite{CMIclustering,kuwahara2019gaussian,PhysRevX.4.031019}. 
First of all, we describe the correlation function $\Cor_\rho(O_X,O_Y)$ can be written as
\begin{align}
\Cor_{\rho_{\beta}}(O_X,O_Y) =\frac{1}{{\cal Z}^2} \tr \br{ e^{-\beta H^{(+)}} O_X^{(0)} O_Y^{(1)}  }, 
\end{align}
where we define the operators $O^{(+)}$, $O^{(0)}$ and $O^{(1)}$ as 
\begin{align}
\label{def_copy_op}
O^{(+)} = O \otimes \hat{1} +  \hat{1} \otimes O , \quad O^{(0)} = O \otimes \hat{1} , \quad  O^{(1)} =  O \otimes \hat{1} -  \hat{1} \otimes O \, , 
\end{align}
by considering a copy of the original Hilbert space. The quantity ${\cal Z}$ is a normalization factor ${\cal Z}^2 =\tr (e^{-\beta H^{(+)}})$($=[\tr(e^{-\beta H})]^2$ in the original Hilbert space). It is simple to observe that the following lemma holds:
\begin{lemma} \label{Cluster_lemma}
{\it For an arbitrary set of operators $\{O_{Z_j}\}_{j=1}^m$, the following relation holds: 
	\begin{align}
		\tr \br{ O_{Z_1}^{(+)} O_{Z_2}^{(+)} \cdots O_{Z_m}^{(+)}   O_X^{(0)} O_Y^{(1)}  } =0
		\label{main:eq_Cluster_lemma}
	\end{align}
	when the subsets $\{Z_1,Z_2,\ldots, Z_m\}$ cannot connect the subsets $X$ and $Y$ [see Fig.~\ref{fig:connection_subset} (a)], i.e., 
	\begin{align} 
		\label{decomp_sets}
		\br{X\cup Z_{i_1} \cup Z_{i_2} \cup \cdots \cup Z_{i_s}} \cap \br{Y\cup Z_{i_{s+1}} \cup Z_{i_{s+2}} \cup \cdots \cup Z_{i_m}} = \emptyset 
	\end{align}
	with $\{i_1,i_2, \ldots, i_m\} = \{1,2,\ldots,m\}$.
}
\end{lemma}

\textit{Proof of Lemma~\ref{Cluster_lemma}.}
Without loss of generality, we let $\{i_1,i_2,\ldots,i_s\}=\{1,2,\ldots, s\}$ and $\{1,2,\ldots,s\}\setminus \{i_1,i_2, \ldots, i_m\} = \{s+1,s+2,\ldots,m\}$. 
We then denote 
\begin{align}
O_{Z_1}^{(+)} O_{Z_2}^{(+)} \cdots O_{Z_s}^{(+)}  = W_1, \quad O_{Z_{s+1}}^{(+)} O_{Z_{s+2}}^{(+)} \cdots O_{Z_m}^{(+)}  = W_2 ,
\end{align}
with
\begin{align}
\Supp(X \cup W_1) \cap \Supp(Y\cup W_2) = \emptyset, 
\end{align}
where the second equation is equivalent to the condition~\eqref{decomp_sets}. 
We then decompose as 
\begin{equation}
\tr \left(O^{(+)}_{Z_1} O^{(+)}_{Z_2}\ldots O^{(+)}_{Z_m} O^{(0)}_X O^{(1)}_Y\right)= \tr \left(W_1 O^{(0)}_X \right)  \tr \left(W_2 O^{(1)}_Y \right).
\end{equation}
Here, the operators $ O^{(+)}_{Z_{i_k}}=O_{Z_{i_k}}\otimes \hat{1}+\hat{1}\otimes O_{Z_{i_k}}$ are symmetric for the exchange of the two Hilbert space while the operator $O^{(1)}_Y=O_Y\otimes \hat{1}-\hat{1}\otimes O_Y$ is antisymmetric, which yields $\tr \left(W_2 O^{(1)}_Y \right)=0$. 
We thus prove the main equation~\eqref{main:eq_Cluster_lemma}.
$\square$

{~}

\hrulefill{\bf [ End of Proof of Lemma~\ref{Cluster_lemma}]}

{~}

\begin{figure}[t]
%  \begin{center}
	\begin{tabular}{c}
		% 1
		\begin{minipage}{0.33\hsize}
			\begin{center}
				\includegraphics[width= \textwidth]{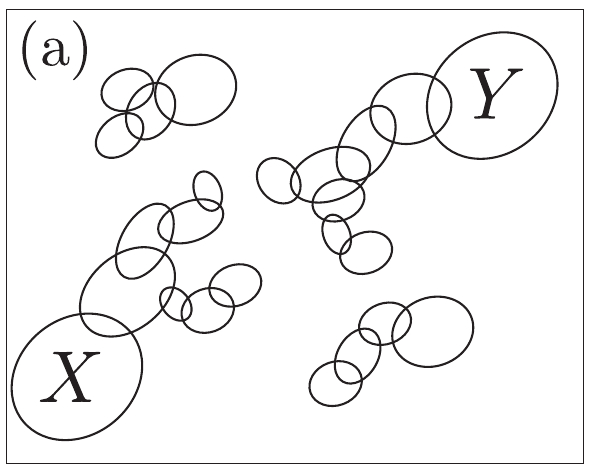}
			\end{center}
		\end{minipage}
		
		% 2
		\begin{minipage}{0.33\hsize}
			\begin{center}
				\includegraphics[width= \textwidth]{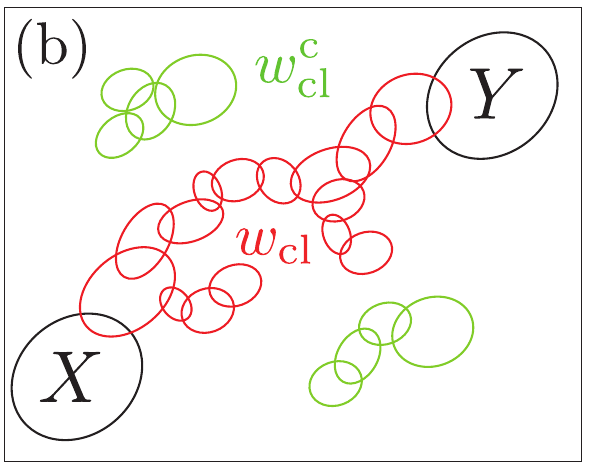}
			\end{center}
		\end{minipage}
		% 3
		\begin{minipage}{0.33\hsize}
			\begin{center}
				\includegraphics[width= \textwidth]{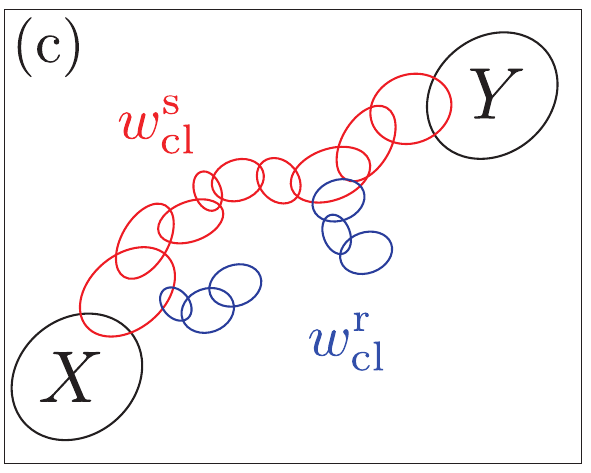}
			\end{center}
		\end{minipage}
	\end{tabular}
	\caption{(a): Schematic picture of unconnected subsets. (b): Connected subsets. (c): Decomposition of the shortest subsets and the rest of the subsets. }
	\label{fig:connection_subset}
\end{figure}

By using Lemma~\ref{Cluster_lemma}, the function $\Cor_{\rho_{\beta}}(O_X,O_Y)$ can be expanded as 
\begin{align}
	\Cor_{\rho_{\beta}}(O_X,O_Y) =\frac{1}{{\cal Z}^2}\sum_{m=0}^\infty \sum_{Z_1,Z_2,\ldots,Z_m: {\rm connected}} \frac{(-\beta )^m}{m!} \tr \br{h_{Z_1}^{(+)} h_{Z_2}^{(+)} \cdots h_{Z_m}^{(+)}O_X^{(0)} O_Y^{(1)}  },
\end{align}
where the summation $\sum_{Z_1, Z_2,\ldots, Z_m: {\rm connected}}$ is taken over the collections $\{X, Z_1,Z_2,\ldots, Z_m,Y\}$ that connect $X$ and $Y$. See Fig.~\ref{fig:connection_subset} (b). Note that $\{Z_s\}_{s=1}^m$ can be the same, e.g., $Z_1=Z_2$. Let us define the operator $\rho_{\rm cl}$ as 
\begin{align}
	\label{correlation_expansion}
	\rho_{\cl}=\frac{1}{{\cal Z}^2} \sum_{m=0}^\infty \sum_{Z_1,Z_2,\ldots,Z_m: {\rm connected}} \frac{(-\beta )^m}{m!} h_{Z_1}^{(+)} h_{Z_2}^{(+)} \cdots h_{Z_m}^{(+)}. 
\end{align}
Using the operator $\rho_{\cl}$, the correlation ${\rm Cor}_{\rho_{\beta}} (O_X, O_Y)$ has the following upper bound: 
\begin{align}
	\label{correlation_up_cluster}
	\left|\Cor_{\rho_{\beta}} (O_X,O_Y) \right| \le 2 \norm{O_X} \cdot \norm{O_Y} \cdot \norm{\rho_{\cl}}_1
\end{align}

To count the summation of $\sum_{Z_1,Z_2,\ldots,Z_m: {\rm connected}}$, we first decompose as 
\begin{align}
	\{ Z_1,Z_2,\ldots,Z_m\} = w_{\rm cl} \oplus w_{\rm cl}^\co,\quad w_{\rm cl}^\co:= \{ Z_1,Z_2,\ldots,Z_m\}\setminus w_{\rm cl},
\end{align}
where $w_{\rm cl}$ is taken such that $w_{\rm cl} \oplus \{X,Y\}$ are connected to each other (See Fig.~\ref{fig:connection_subset} (b)), in other words, 
\begin{align}
	Z \cap Z' = \emptyset \for Z \in w_{\rm cl}  ,\quad Z' \in w^\co_{\rm cl}.
\end{align}
%Furthermore, we denote the notation of $w^\ast_{\rm cl})$ by removing the duplicating subsets from $w_{\rm cl}$; for example, for $w^\ast_{\rm cl}=\{Z,Z,Z',Z'',Z''\}$, we have $w^\ast_{\rm cl}=\{Z,Z',Z''\}$.
We define all the sets of $w_{\rm cl}$ and $w_{\rm cl}^{\rm c}$ as $\mathcal{G}_{\rm cl}$ and $\mathcal{G}^{\rm c}_{\rm cl}$, respectively. We also define the total support of the $w_{\rm cl}$ and $w_{\rm cl}^{\rm c}$ as $V_{w_{\rm cl}}$ and $V_{w_{\rm cl}^{\rm c}}$, respectively. We then obtain 
\begin{align}
	\rho_{\rm cl} %&= {1\over {\cal Z}^2} \sum_{m=0}^{\infty }{(-\beta )^m \over m! }\sum_{Z_1,Z_2,\ldots,Z_m: {\rm connected}} h_{Z_1}^{(+)} h_{Z_2}^{(+)} \cdots h_{Z_m}^{(+)}\\
	&= {1\over {\cal Z}^2} \sum_{m=0}^{\infty }{(-\beta )^m \over m! }\sum_{s=0}^{m} {m! \over (m-s)! s! } \sum_{\{ Z_1 , \cdots , Z_s\} = w_{\rm cl} \in \mathcal{G}_{\rm cl}}\!\!\!\! \tilde{h}(Z_1, \cdots , Z_s) \sum_{w_{\rm cl}^{\rm c} \in \mathcal{G}^{\rm c}_{\rm cl}}\sum_{Z_{s+1}, \cdots , Z_m  \atop \in w_{\rm cl}^{\rm c}}\!\!\!\!\!\! h^{(+)} (Z_{s+1}) h^{(+)}(Z_{s+2})\cdots  h^{(+)}(Z_{m}) \nonumber \\
	&= {1\over {\cal Z}^2} \sum_{s=0}^{\infty }{(-\beta )^s \over s! } \sum_{w_{\rm cl} \in \mathcal{G}_{\rm cl}} \tilde{h}(Z_1, \cdots , Z_s) e^{-\beta H^{(+)}_{V_{w^\co_{\rm cl}}}} .\label{decomp_cl_and_cl-co}
	%&= \sum_{m=0}^\infty \sum_{Z_1,Z_2,\ldots,Z_m: {\rm connected}} \frac{(-\beta )^m}{m!} h_{Z_1}^{(+)} h_{Z_2}^{(+)} \cdots h_{Z_m}^{(+)} \\
	%&= \sum_{w_{\rm cl} \in \mathcal{G}_{\rm cl}}\frac{(-\beta )^{|w_{\rm cl}|}}{|w_{\rm cl}|!}  h_{w_{\rm cl}}^{(+)} 
\end{align}
Here, we have defined 
\begin{align}
	\tilde{h}(Z_1 , \cdots , Z_s) & =\sum_{P} h^{(+)} (Z_{P_1})\cdots  h^{(+)}(Z_{P_s}) \, , 
\end{align} 
where $P$ implies taking all combinations of different elements for $Z$. In particular, when $Z$ are all different, it means taking a permutation.
%where we define $h_{w_{\rm cl}}^{(+)} = \prod_{Z\in w_{\rm cl}}h_Z^{(+)}$.
%In the second equation above, we used the fact that for fixed $w_{\rm cl}= \{ Z_1,Z_2,\ldots,Z_s\}$ and $w_{\rm cl}^\co= \{ Z_{s+1},Z_{s+2},\ldots,Z_m\}$, the number of patterns of $\{ Z_1,Z_2,\ldots,Z_m\}$ to give the same $w_{\rm cl}$ and $w_{\rm cl}^\co$ is 
%\begin{align}
%\binom{m}{s} =\frac{m!}{s!(m-s)!} =\frac{m!}{|w_{\rm cl}|!(m-s)!}  . 
%\end{align}
From the expression (\ref{decomp_cl_and_cl-co}), we obtain
\begin{align}
	\| {\rho_{\cl}}\|_1 &\le 
	{1\over {\cal Z}^2} \sum_{s=0}^{\infty }{\beta^s \over s! } \sum_{ w_{\rm cl} \in \mathcal{G}_{\rm cl}} \| \tilde{h}(Z_1, \cdots , Z_s) \| \,  \tr \Bigl[ e^{-\beta H^{(+)}_{V_{w^\co_{\rm cl}}}} \Bigr] 
	%\frac{1}{Z^2}  \sum_{w_{\rm cl} \in \mathcal{G}_{\rm cl}} \frac{\beta^{|w_{\rm cl}|}}{|w_{\rm cl}|!}  \norm{h_{w_{\rm cl}}^{(+)}} \tr\br{ e^{-\beta H^{(+)}_{V_{w^\co_{\rm cl}}}}}
	\le 
	\sum_{s=0}^{\infty }{\beta^s \over s! } \sum_{ w_{\rm cl} \in \mathcal{G}_{\rm cl}} e^{2\beta g u| V_{w_{\rm cl}}|} \| \tilde{h}(Z_1, \cdots , Z_s) \| \, . 
	%\sum_{w_{\rm cl} \in \mathcal{G}_{\rm cl}}e^{2\beta g | V_{w_{\rm cl}}|} \frac{\beta^{|w_{\rm cl}|}}{|w_{\rm cl}|!}   \norm{h_{w_{\rm cl}}^{(+)}} 
\end{align}
Here, we use the Golden-Thompson inequality as follows:
\begin{align}
	\tr\Bigl[ e^{-\beta H^{(+)}_{V_{w^\co_{\rm cl}}}} \Bigr]&= \tr \Bigl[ e^{-\beta \br{ H^{(+)} + H^{(+)}_{V_{w^\co_{\rm cl}}}- H^{(+)}}} \Bigr]
	\le  \tr \Bigl[ e^{-\beta H^{(+)}  }  e^{-\beta  \br{ H^{(+)}_{V_{w^\co_{\rm cl}}}-H^{(+)}}} \Bigr] \le {\cal Z}^2  e^{2\beta \|  H_{V^\co_{w_{\rm cl}}}-H \| }  \nonumber \\
	&\le {\cal Z}^2  e^{2\beta g u| V_{w_{\rm cl}}|}\le {\cal Z}^2  e^{2\beta gu ks}
\end{align}
where we use the condition~\eqref{extensive_cond} in the last inequality. 
From the $k$-locality of the Hamiltonian, we have $|V_{w_{\rm cl}}| \le k| w_{\rm cl} |= k s $. 
%$|V_{w_{\rm cl}}|=|V_{w_{\rm cl}^\ast}| \le k| w_{\rm cl}^\ast|$. 

We divide each element in the set $w_{\rm cl}$ into the two subsets $w_{\rm cl}^{\rm s}$ and $w_{\rm cl}^{\rm r}$, where $w_{\rm cl}^{\rm s}$  is a set representing the shortest path with the length $s'$ and $w_{\rm cl}^{\rm r}$ is a set for the rest part with the length $(s-s')$. See Fig.~\ref{fig:connection_subset} (c). We thus obtain 
\begin{align}
	\| \rho_{\cl} \|_1& \le \sum_{s=0}^{\infty} { (\beta e^{2\beta g u k })^s\over s! } \sum_{s'=1}^s {s! \over s' ! (s-s')!} \sum_{\{Z_1, \cdots , Z_{s'}\} =w_{\rm cl}^{\rm s}} \| \tilde{h} (Z_1, \cdots ,Z_{s'})\|' \sum_{\{Z_{s'+1}, \cdots , Z_{s}\} =w_{\rm cl}^{\rm r}} \| \tilde{h} (Z_{s'+1}, \cdots ,Z_{s})\|'  \nonumber \\
	& = \sum_{s=0}^{\infty} (\beta e^{2\beta g u k })^s\sum_{s'=1}^s
	\sum_{w_{\rm cl}^{\rm s}} { \| \tilde{h} (Z_1, \cdots ,Z_{s'})\|'  \over  s'! }
	\sum_{w_{\rm cl}^{\rm r}} { \| \tilde{h} (Z_{s'+1}, \cdots ,Z_{s})\|' \over (s -s')! }\, , \label{rhocl1wclsr}
	%
	%\sum_{w_{\rm cl} \in \mathcal{G}_{\rm cl}}e^{2\beta g k | w_{\rm cl}^\ast|} \frac{\beta^{|w_{\rm cl}|}}{|w_{\rm cl}|!}\prod_{Z\in  w_{\rm cl}} 2 \norm{h_Z} , 
	%\norm{\rho_{\cl}}\le \sum_{w_{\rm cl} \in \mathcal{G}_{\rm cl}}e^{2\beta g k | w_{\rm cl}^\ast|} \frac{\beta^{|w_{\rm cl}|}}{|w_{\rm cl}|!}\prod_{Z\in  w_{\rm cl}} 2 \norm{h_Z} , 
\end{align}
where 
\begin{align}
	\| \tilde{h} (Z_1, \cdots ,Z_{s'})\|' &: =\sum_{P} \| h_{Z_{P_1}}^{(+)}\| \cdots \| h_{Z_{P_{s'}}}^{(+)}\| \le
	2^{s'} \sum_{P} \| h_{Z_{P_1}} \| \cdots \| h_{Z_{P_{s'}}}\|  \, . \label{1sp}
\end{align}
%Note that we use $\| h_Z^{(+)} \| \le 2 \norm{h_Z}$ from the definition~\eqref{def_copy_op}. 
Similarly to (\ref{1sp}), $ \| \tilde{h} (Z_{s'+1}, \cdots ,Z_{s})\|'$ is defined. Below, we evaluate the terms on $w_{\rm cl}^{\rm s}$ and $w_{\rm cl}^{\rm r}$ in (\ref{rhocl1wclsr}). To this end, the following lemma is useful. 
\begin{lemma} \label{productlemma}
	{\it For $\alpha >D$, 
		the quantity $J_{i,j}$ defined in (\ref{def_short_range_long_range}) satisfies the following inequality
		\begin{align}
			\sum_{j \in \Lambda} J_{i,j} J_{j,k} \le {g^2 u \over (1+d_{i,k})^{\alpha}} \, , \label{twice}
		\end{align}
		where $u=2^{\alpha} \sum_{j \in \Lambda} (1+d_{i,j})^{-\alpha}$. The iterative use of (\ref{twice}) leads to $\left[ {\bm J}^{\ell}\right]_{i,k} \le {g^{\ell}u^{\ell -1} / (1+ d_{i,k})^{\alpha}}$.
	}
\end{lemma}

\textit{Proof of Lemma~\ref{productlemma}.}
\begin{align}
	\sum_{j \in \Lambda }J_{i,j} J_{j,k} &=g^2 \sum_{j \in \Lambda }{1\over (1+d_{i,j})^\alpha }{1\over (1+d_{j,k})^\alpha } \le {g^2 \over (1+ d_{i,k})^{\alpha}} \sum_{j \in \Lambda }{ (2 + d_{i,j} + d_{j,k})^\alpha \over (1+d_{i,j})^\alpha }{1\over (1+d_{j,k})^\alpha } .
\end{align}
Note $(x+y)^{\alpha} \le 2^{\alpha -1} (x^{\alpha} + y^{\alpha})$ for $\alpha >1$. Then we have 
\begin{align}
	\sum_{j \in \Lambda } J_{i,j} J_{j,k} &\le {g^2 2^{\alpha -1} \over (1+ d_{i,k})^{\alpha}}
	\sum_{j \in \Lambda } {1 \over (1+d_{i,j})^\alpha } + {1 \over (1+d_{j,k})^\alpha }
	= {g^2 u \over (1+ d_{i,k})^{\alpha}} .
\end{align} 
The finiteness of $u$ is guaranteed by the condition $\alpha >D$. 

{~}

\hrulefill{\bf [ End of Proof of Lemma~\ref{productlemma}]}

{~}

Let us first consider the terms on $w_{\rm cl}^{\rm s}$ in (\ref{rhocl1wclsr}).
\begin{align}
	\sum_{w_{\rm cl}^{\rm s}} { \| \tilde{h} (Z_1, \cdots ,Z_{s'})\|'  \over  s'! }
	&\le {2^{s'}} \sum_{i_1 \in X}\sum_{i_2 \in \Lambda } \sum_{Z_1 \ni \{ i_1, i_2\}}\sum_{i_3 \in \Lambda }\sum_{Z_2 \ni \{ i_2 , i_3\}} \cdots \sum_{i_{s'} \in \Lambda } \sum_{i_{s' +1} \in Y} \sum_{Z_{s'}\ni \{i_{s'}, i_{s' +1}\}}  \| h_{Z_{1}} \| \cdots \| h_{Z_{s'}}\| \nonumber \\
	&={2^{s'}} \sum_{i_1 \in X}\sum_{i_2 \in \Lambda } \cdots  \sum_{i_{s'} \in \Lambda } \sum_{i_{s' +1} \in Y} J_{i_1,i_2}\cdots J_{i_{s'},i_{s'+1}}
	\le {u^{-1} (2 g u )^{s'}}|X| \,|Y| {1 \over (1 + d_{X,Y})^{\alpha }} \, .
\end{align}

\begin{figure}[t]
	\begin{tabular}{c}
		%        \begin{center}
			\includegraphics[width= 0.55\textwidth]{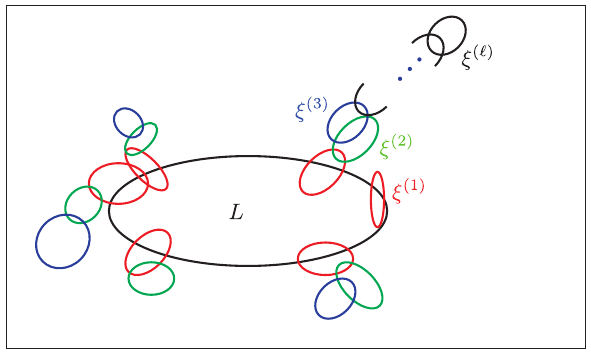}
			%        \end{center}
	\end{tabular}
	\caption{Schematic picture of the layered structure}
	\label{fig:layer}
\end{figure}
We next consider the terms of $w_{\rm cl}^{\rm r}$ in (\ref{rhocl1wclsr}). We define $n:=s-s'$ and $L:=X \cup Y \cup \omega_{\rm cl}^{\rm s} $ for a simplicity. In addition, let us denote $Z_{s+m}$ by $\xi_{m}\,(m=1,\cdots , n)$ for notational simplicity. For a given set $\xi_L:=\{\xi_m \}_{m=1}^n$, we decompose it as follows.
\begin{align}
	\xi_L &= \xi^{(1)}\oplus\xi^{(2)}\oplus \cdots \oplus \xi^{(\ell)} \, , ~~~~~ 1 \le \ell \le n \, 
\end{align}
where the set $\xi^{(j)}$ has the overlap $\xi^{(j-1)}$ with setting $\xi^{(0)}:=L$. Hence, $\xi^{(j)}$ is a achievable set from $L$ by the $j$-step connection. See the schematic in Fig.\ref{fig:layer}. Let $n_j=|\xi^{(j)} |$ where $n_j \ge 1$. We take account of all possible combination of this structure:
\begin{align}  
	\sum_{w_{\rm cl}^{\rm r}} { \|\tilde{h} (\xi_1, \cdots , \xi_n) \| \over n!}
	&={2^n \over n!} \sum_{(\xi_1 , \cdots , \xi_n)} \| h_{\xi_1} \| \cdots\| h_{\xi_n} \| \nonumber \\
	&={2^n \over n!}\sum_{\ell=1}^{n} \sum_{n_1 +\cdots + n_{\ell}=n , n_i \ge 1} {n! \over n_1 ! \cdots n_{\ell}!} \nonumber \\
	&\times \left( \sum_{\xi^{(1)}_i \in \xi^{(1)}|L} \| h_{\xi^{(1)}_i}\| \right)^{n_1} \left( \sum_{\xi^{(2)}_i \in \xi^{(2)}|\{L, \xi^{(1)}\}} \| h_{\xi^{(2)}_i}\| \right)^{n_2} \cdots \left( \sum_{\xi^{(\ell)}_i \in \xi^{(\ell)}|\{L, \xi^{(1)},\cdots , \xi^{(\ell -1)}\}} \| h_{\xi^{(\ell)}_i}\| \right)^{n_{\ell}} \, . \label{wclr}
\end{align}
Here, $\xi^{(j)}|\{L, \xi^{(1)},\cdots , \xi^{(j-1)}\} $ implies taking a set $\xi^{(j)}$ for fixed set $\{L, \xi^{(1)},\cdots , \xi^{(j-1)}\}$. We first upper-bound the most right term in (\ref{wclr}):
\begin{align}
	\sum_{\xi^{(\ell)}_i \in \xi^{(\ell)}|\{L, \xi^{(1)},\cdots , \xi^{(\ell -1)}\}} \| h_{\xi^{(\ell)}_i}\| 
	&\le \sum_{j \in \{ \xi^{(\ell -1) }\}} \sum_{\xi^{(\ell)}_i \ni j } \| h_{\xi^{(\ell)}_i}\| \le (k n_{\ell -1}) g u \, , 
\end{align}
where we use $|\xi^{(\ell -1)}|=n_{\ell -1}$ and $\sum_{\xi \ni i} \| h_{\xi}\| \le g u $. Iterative use of this computation down to the summation over $\xi_i^{(1)}$, we obtain the following bound 
\begin{align}
	\sum_{w_{\rm cl}^{\rm r}} { \|\tilde{h} (\xi_1, \cdots , \xi_n) \| \over n!}
	&\le {(2 k gu)^n }\sum_{\ell=1}^{n} \sum_{n_1 +\cdots + n_{\ell}=n \atop  n_i \ge 1} {1 \over n_1 ! \cdots n_{\ell}!}(|L|/k)^{n_1} n_1^{n_2} \cdots n_{\ell -1}^{n_{\ell}} \, . 
\end{align}
We here use the relation $n_j! \ge (n_j/e)^{n_j}$ and $(1 + y/x)^x \le e^{y}$ for $x>0$ and $y>0$ to get
\begin{align}
	\sum_{w_{\rm cl}^{\rm r}} { \|\tilde{h} (\xi_1, \cdots , \xi_n) \| \over n!} &\le {(2 k gu )^n } \sum_{\ell=1}^{n} \sum_{n_1 +\cdots + n_{\ell}=n \atop  n_i \ge 1} e^{n_1} (|L|/k/n_1)^{n_1}e^{n_2}(n_1/n_2)^{n_2}\cdots e^{n_{\ell}} (n_{\ell-1}/n_{\ell})^{n_\ell} \nonumber \\
	&\le {(2 k gu e)^n } \sum_{\ell=1}^{n} \sum_{n_1 +\cdots + n_{\ell}=n \atop  n_i \ge 1}  (1+|L|/k/n_1)^{n_1}(1+ n_1/n_2)^{n_2}\cdots (1+ n_{\ell-1}/n_{\ell})^{n_\ell} \nonumber \\
	&\le {(2 k gu e^2)^n } e^{|L|/k} \sum_{\ell=1}^{n} \sum_{n_1 +\cdots + n_{\ell}=n \atop  n_i \ge 1}   \nonumber \\
	&= {1\over 2} {(4 k gu e^2)^n } e^{|L|/k} \nonumber \\
	&\le {1\over 2}(4 k gu e^2)^n e^{(|X|+|Y| + s' k)/k} , 
\end{align}
where at the last line we have used the relation $|L| < |X|+ |Y|+ s' k$. At the third line, we use the relation,
$ \sum_{\ell=1}^{n} \sum_{n_1 +\cdots + n_{\ell}=n \atop n_i \ge 1}=2^{n-1}$. Finally we sum over $s'$ to get the following bound.
\begin{align}
	\|\rho_{\rm cl} \|_1 &\le {u^{-1}\over 2} |X|\, |Y| {e^{(|X| + |Y|)/k}\over (1 + d_{X,Y})^{\alpha}} \sum_{s=0}^{\infty} (4\beta g u k e^2 \, e^{2\beta g u k} )^s \sum_{s'=1}^s {(1/2ke)^{s'} } \nonumber \\
	&\le  u^{-1} |X|\, |Y| {e^{(|X| + |Y|)/k}\over (1 + d_{X,Y})^{\alpha}} {1\over 1- 4\beta g u k e^2 \, e^{2\beta g u k } } \, . 
\end{align}
This completes the proof of the Theorem \ref{main_thm_1}.

%%%%%%%%

\section{Power-Law Fitting of 1D Bilinear Fermion Systems}
To clarify the increasing trend of the mutual information, we employ power-law fitting to the 1D long-range bilinear fermion systems characterized by the following Hamiltonian
\begin{align}
	H = - \sum_{\substack{i,j = 1 \\ i \neq j}}^{N} \frac{t_{i,j}}{\dist_{i,j}^{\alpha}} (c_{i}^{\dagger} c_{j} + c_{j}^{\dagger} c_{i}). \label{Model_1_supp}
\end{align}
Here, $t_{i,j}$ denotes the hopping parameter of order $\mathcal{O}(1)$, $N$ represents the system size, and $c_{i}$ and $c_{i}^{\dagger}$ correspond to the annihilation and creation operators, respectively, for a spinless fermion at site $i$.
We numerically calculate the mutual information $\mathcal{I}_{\rho_{\beta}}(A:B)$ at $\beta = 2$ between subsystems $A$ and $B$, where $A$ denotes the first $N / 2$ sites and $B$ the remaining $N / 2$ sites.
We vary the system size $N$ from 100 to 1000 in increments of 100 and examine the power $\alpha$ in Eq.~\eqref{Model_1_supp} from 0.5 to 1.4 in increments of 0.1.

Fig.~\ref{fig:exp}(a) illustrates the average mutual information computed over 20,000 random values of $t_{i,j}$ sampled from the interval $[0,1]$. The results reveal that the mutual information increases following the behavior of $N^{2 - 2\alpha}$ for $\alpha < 1$, while saturating when $\alpha > 1$.

At $\alpha = 1$, the numerical data indicates a slight upward trend. 
To analyze this behavior in greater detail, we plot the average mutual information as a function of $\ln N$, computed from 100,000 samples (see Fig.~\ref{fig:exp}(b)). 
The analysis shows that at $\alpha = 1$, the mutual information exhibits a logarithmic growth with respect to $N$.
A linear regression, performed using MATLAB's \texttt{polyfit} function, confirms that the data (circles in Fig.~\ref{fig:exp}(b)) closely follows a linear dependence on $\ln N$, described by $\mathcal{I}_{\rho_{\beta}}(A:B) = 0.0575 \ln N + 0.2082$ (red line in Fig.~\ref{fig:exp}(b)). 
It is worth noting that Theorem 1 affirms the validity of the thermal area law for $\alpha > 1$, thereby excluding $\alpha = 1$ from this regime.

To further clarify the behavior for $\alpha < 1$ and $\alpha > 1$, we fit the mutual information data from Fig.~\ref{fig:exp}(a) to a power-law form, $\mathcal{I}_{\rho_{\beta}}(A:B) \approx c N^{\Delta}$, and extract the coefficient $c$ and exponent $\Delta$ via numerical fitting. 
In Fig.~\ref{fig:exp}(c), we plot the resulting exponent $\Delta$ as a function of the interaction power $\alpha$ in Eq.~\eqref{Model_1_supp}. 
The power-law fitting was performed using MATLAB's \texttt{polyfit} function, applied to the mutual information data for $N$ in the ranges $100 \leq N \leq 500$ (circles in Fig.~\ref{fig:exp}(c)) and $600 \leq N \leq 1000$ (diamonds in Fig.~\ref{fig:exp}(c)). 
For $\alpha < 1$, the extracted exponent $\Delta$ closely follows the linear relation $\Delta = 2 - 2\alpha$, suggesting that the mutual information scales as $N^{2 - 2\alpha}$. 
In contrast, for $\alpha > 1$, the exponent $\Delta$ assumes a small value consistent with the thermal area law. 
Due to system size limitations, the observed exponent is not exactly zero, particularly when $\alpha$ is close to 1. 
However, it can be inferred that the exponent approaches zero as the system size increases because Theorems 1 and 2 guarantee the thermal area law in the 1D case. 
Supporting this inference, the data for the range $600 \leq N \leq 1000$ more closely aligns with the behavior $\Delta = 2 - 2\alpha$ for $\alpha < 1$ and $\Delta = 0$ for $\alpha > 1$, compared to the data for $100 \leq N \leq 500$. 
This suggests that the behavior becomes more evident as the system size increases. 
This supports the conclusion that the transition point occurs at $\alpha = 1$ and that the thermal area law is valid for $\alpha > 1$.

\begin{figure}[h]
	\centering
	\includegraphics[width=\textwidth]{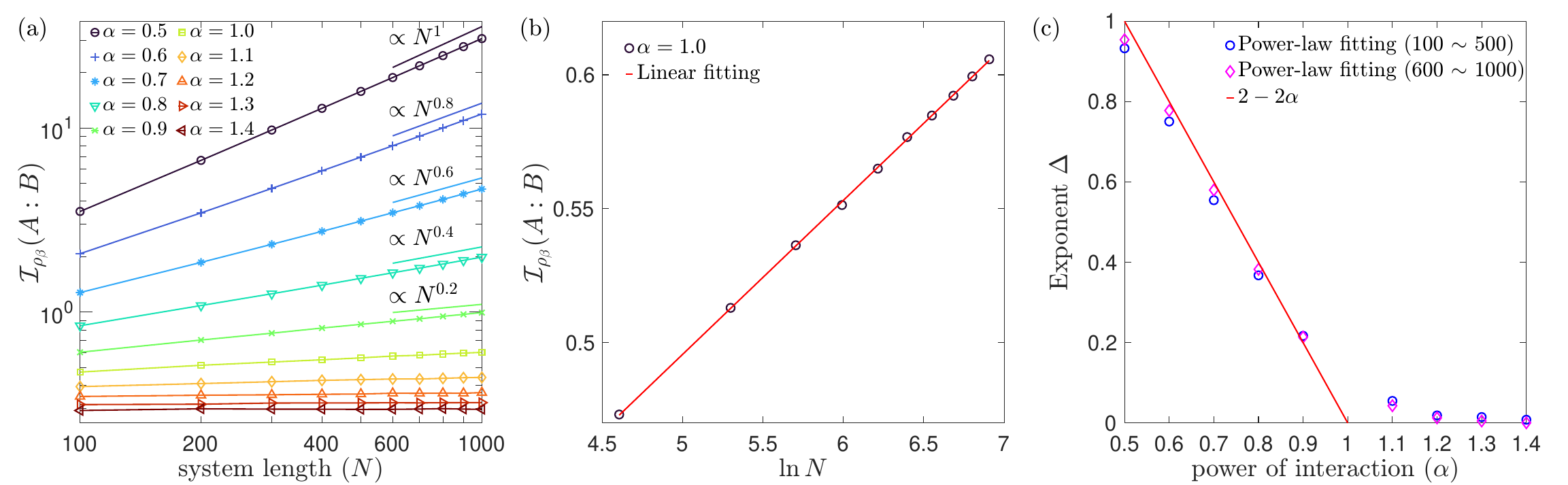}
	\caption{(a) The mutual information $\mathcal{I}_{\rho_{\beta}}(A:B)$ between subsystem $A$ (the first $N/2$ sites) and subsystem $B$ (the remaining $N/2$ sites) is computed. The system size $N$ varies from 100 to 1000 in increments of 100, while the interaction power $\alpha$ varies from 0.5 to 1.4 in increments of 0.1. The data represents the average mutual information over 20,000 samples, with the hopping parameters $t_{i,j}$ randomly selected from the interval $[0,1]$. For $\alpha < 1$, a line proportional to $N^{2 - 2\alpha}$ is included in the figure for comparison.
		(b) The average mutual information at $\alpha = 1.0$, calculated from 100,000 samples, is plotted as circles with the $x$-axis representing the natural logarithm of the system size, $\ln N$. The red line corresponds to the linear fit, given by $\mathcal{I}_{\rho_{\beta}}(A:B) = 0.0575 \ln N + 0.2082$.  
		(c) The circles (diamonds) represent the exponents $\Delta$, obtained by fitting the average mutual information from panel (a) to a power-law, $\mathcal{I}_{\rho_{\beta}}(A:B) \approx c N^{\Delta}$, over the ranges $N = 100$ to $N = 500$ (resp. $N = 600$ to $N = 1000$).
		The red line represents $\Delta = 2 - 2\alpha$. 
		For $\alpha < 1$, the exponent $\Delta$ closely follows the line $2 - 2\alpha$, whereas for $\alpha > 1$, $\Delta$ stays near zero.
		This behavior becomes more pronounced as the system size increases.
	}
	\label{fig:exp}
\end{figure}

\section{Exact Diagonalization Results for Long-Range Heisenberg Chains}
We now consider a general 1D long-range Heisenberg chain described by the Hamiltonian
\begin{align}
	H = \sum_{1 \leq i<j \leq N} \frac{a_{i,j}}{\dist_{i,j}^{\alpha}} {\bm S}_{i} \cdot {\bm S}_{j}, \label{Model_3}
\end{align}
where $a_{i,j}$ is a parameter of order $\orderof{1}$, $N$ is the number of sites, and ${\bm S}_{i} = (S_{i}^{x},S_{i}^{y},S_{i}^{z})$ is a spin-1/2 operator at site $i$.
To support our argument that the thermal area holds for $\alpha > (D+1) / 2$ and this criterion is optimal, we compute the mutual information $\mathcal{I}_{\rho_{\beta}}(A:B)$ between subsystems $A$ and $B$ at $\beta = 2$, with $A$ representing the first $N/2$ sites and $B$ the remaining $N/2$ sites, by using the exact diagonalization method. 
We randomly assign $a_{i,j}$ from the interval $[0,1]$, and then compute the average over 10000 samples. 
Fig.~\ref{fig:ED} shows that the mutual information increases with the system size when $\alpha < 1$, while it does not exhibit an increase for $\alpha > 1$.

\begin{figure}[h]
	\centering
	\includegraphics[width=0.45\textwidth]{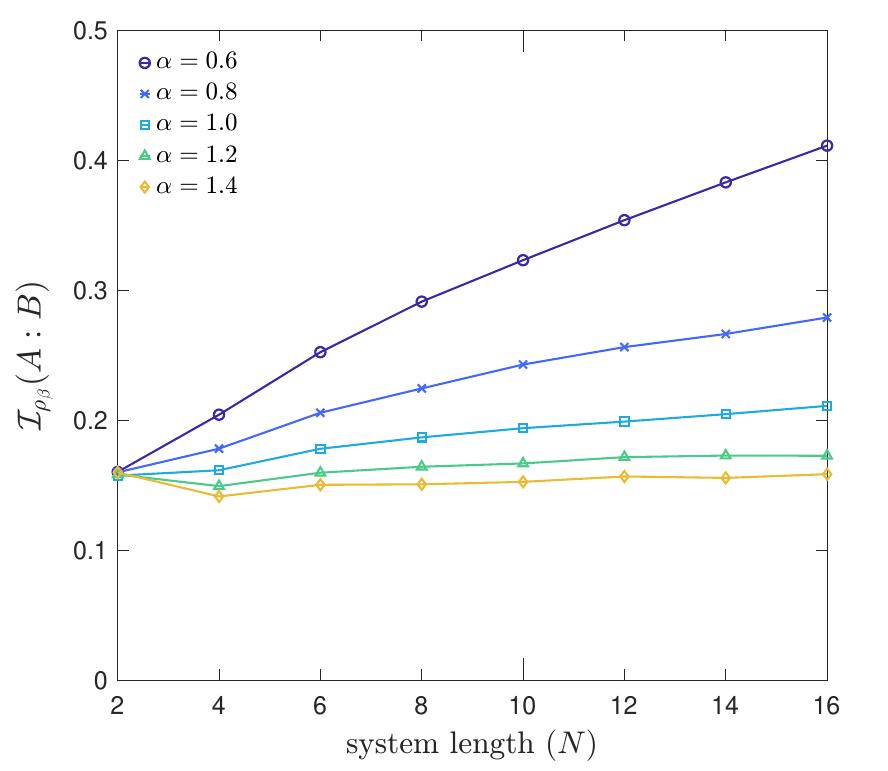}
	\caption{Exact diagonalization results: the averaged mutual information $\mathcal{I}_{\rho_{\beta}}(A:B)$ over 10000 samples of the long-range Heisenberg chain~\eqref{Model_3} with random numbers $a_{i,j}$ assigned from the interval $[0,1]$. Here $N$ is the number of sites, $\beta = 2$, and $A$ and $B$ are the first $N/2$ sites and the remaining $N/2$ sites, respectively.}
	\label{fig:ED}
\end{figure}

\section{Numerical Investigation of Clustering Properties in 2D Long-Range Interacting Spin Systems}
In this section, we present numerical analyses of the clustering properties of correlation functions in two-dimensional long-range interacting spin systems. These results extend the findings for bilinear fermion systems discussed in the main text, providing further evidence that clustering properties universally hold even in nonintegrable systems. In particular, we show that both the long-range Heisenberg model and the long-range XX model exhibit clustering properties consistently in the regime $\alpha < D$.

We consider a rectangular lattice of dimensions $L_x \times L_y$, where the lattice sites are indexed by integer coordinates $\boldsymbol{r} = (x, y)$ with $1 \leq x \leq L_x$ and $1 \leq y \leq L_y$. The distance between two lattice sites $\boldsymbol{r}_{1} = (x_{1}, y_{1})$ and $\boldsymbol{r}_{2} = (x_{2}, y_{2})$ is given by $|\boldsymbol{r}_{1} - \boldsymbol{r}_{2}| = \sqrt{(x_{1} - x_{2})^{2} + (y_{1} - y_{2})^{2}}$.

\subsection{2D Long-Range Heisenberg Model}
We study the 2D long-range Heisenberg model characterized by the Hamiltonian:  
\begin{align}  
	H_{\mathrm{Hei}} = \sum_{\substack{\boldsymbol{r}_{1}, \boldsymbol{r}_{2} \\ \boldsymbol{r}_{1} \neq \boldsymbol{r}_{2}}}  
	\frac{1}{2 |\boldsymbol{r}_{1}-\boldsymbol{r}_{2}|^{\alpha}} \boldsymbol{S}_{\boldsymbol{r}_{1}} \cdot \boldsymbol{S}_{\boldsymbol{r}_{2}},
\end{align}  
where $\boldsymbol{S}_{\mathbf{r}} = (S_{\mathbf{r}}^{x}, S_{\mathbf{r}}^{y}, S_{\mathbf{r}}^{z})$ represents the Pauli spin vector at lattice site $\boldsymbol{r}$.

We calculate the two-point correlation function $\langle \boldsymbol{S}_{\boldsymbol{i}} \cdot \boldsymbol{S}_{\boldsymbol{i} + \boldsymbol{r}} \rangle = \mathrm{Tr}(\rho_{\beta} \, \boldsymbol{S}_{\boldsymbol{i}} \cdot \boldsymbol{S}_{\boldsymbol{i} + \boldsymbol{r}})$ for the Gibbs state $\rho_{\beta} = e^{-\beta H_{\mathrm{Hei}}} / \mathrm{Tr}(e^{-\beta H_{\mathrm{Hei}}})$, between two sites $\boldsymbol{i}$ and $\boldsymbol{i} + \boldsymbol{r}$. The results are illustrated in Fig.~\ref{fig:2D_Heisenberg}.  

Two types of geometries are considered. The first is a rectangular lattice with dimensions $L_x = 16$ and $L_y = 5$. In Fig.~\ref{fig:2D_Heisenberg}(a), the correlation function is computed at $\beta = 0.1$ by fixing $\boldsymbol{i} = (3,3)$ and varying $\boldsymbol{r} = (r,0)$ for $r = 1, 2, \cdots, 10$. The second geometry is a square lattice with $L_x = L_y = 9$. In Fig.~\ref{fig:2D_Heisenberg}(b), we fix $\boldsymbol{i} = (2,5)$ and calculate the correlation function for $\boldsymbol{r} = (r,0)$ with $r = 1, 2, \cdots, 6$ under the same temperature, $\beta = 0.1$. In both setups, the correlation functions are computed within the bulk regions, focusing on the middle layer of the lattice along the $y$-direction, while excluding the edge sites.

By varying $\alpha$ between 0.9 and 1.8, we find that the absolute value $\abs{\langle \boldsymbol{S}_{\boldsymbol{i}} \cdot \boldsymbol{S}_{\boldsymbol{i} + \boldsymbol{r}} \rangle}$ of the correlation function is consistently bounded above by $C / |\boldsymbol{r}|^{\alpha}$ (where $C$ is a constant), depending on the separation distance $|\boldsymbol{r}|$. These results confirm that the clustering property robustly holds even for $\alpha < D$.

We employed the exponential tensor renormalization group (XTRG) to construct the matrix product operator (MPO) representation of the Gibbs state, which was subsequently utilized to compute the correlation functions. For 2D systems, the MPO was implemented using the standard snakelike mapping, which transforms the 2D lattice into an effective 1D structure by mapping $(x, y)$ to a 1D index defined as $x + L_{x} \times (y - 1)$. To improve computational efficiency, we incorporated SU(2) symmetry, setting the MPO bond dimension to 300 SU(2) multiplets, corresponding to approximately 1150 individual states.

\begin{figure}[h]
	\centering
	\includegraphics[width=0.7\textwidth]{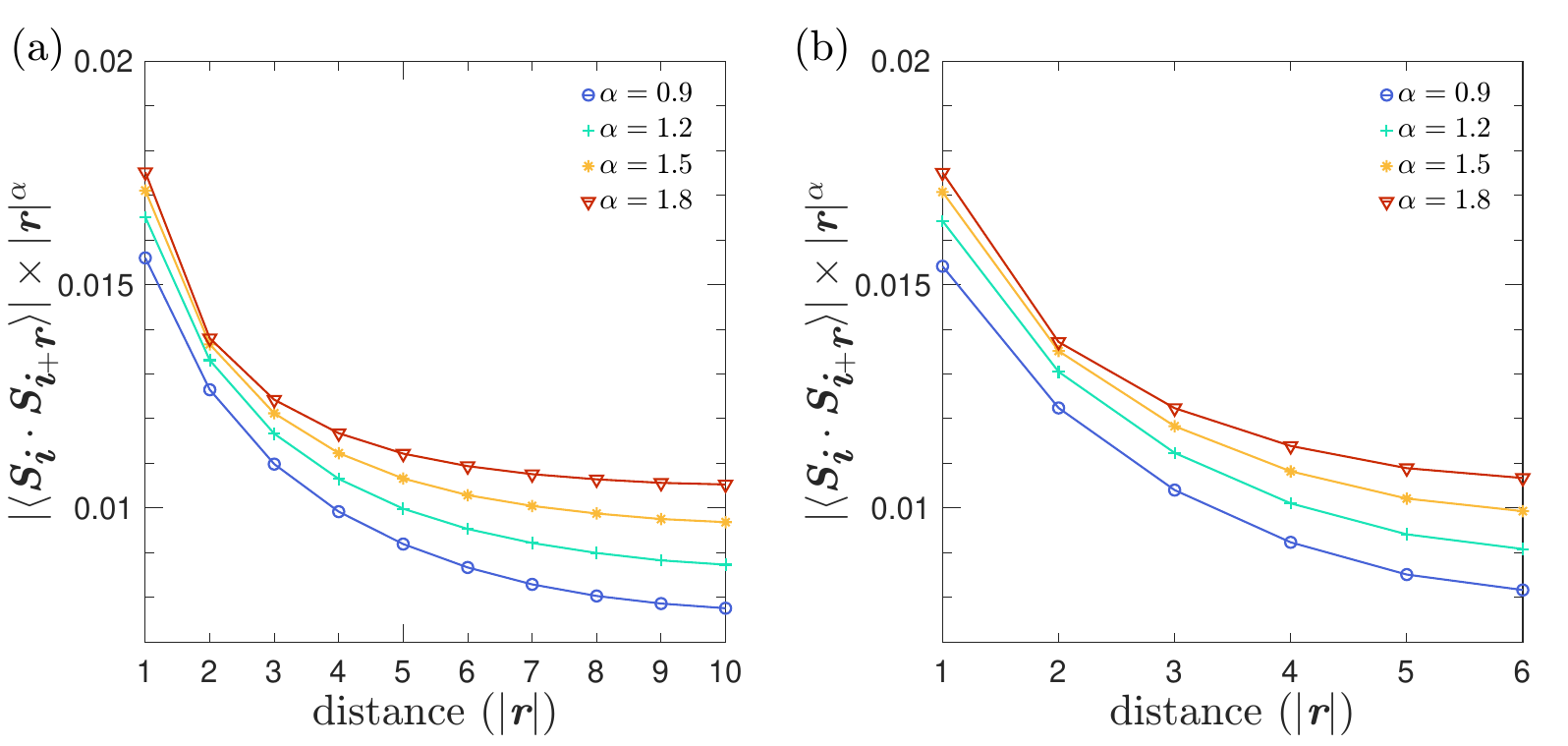}
	\caption{Scaled two-point correlation functions $\abs{\langle \boldsymbol{S}_{\boldsymbol{i}} \cdot \boldsymbol{S}_{\boldsymbol{i} + \boldsymbol{r}} \rangle} \times \abs{\boldsymbol{r}}^{\alpha}$ for the 2D long-range Heisenberg model. 
		(a) Results for a rectangular lattice ($L_x = 16, L_y = 5$) with $\boldsymbol{i} = (3,3)$, varying $\boldsymbol{r} = (r,0)$ for $r = 1, 2, \cdots, 10$ at $\beta = 0.1$.  
		(b) Results for a square lattice ($L_x = L_y = 9$) with $\boldsymbol{i} = (2,5)$, varying $\boldsymbol{r} = (r,0)$ for $r = 1, 2, \cdots, 6$ at $\beta = 0.1$.  
		The plots illustrate the dependence of the scaled correlation function on $|\boldsymbol{r}|$, confirming the clustering property for $\alpha < D$.}
	\label{fig:2D_Heisenberg}
\end{figure}

\subsection{2D Long-Range XX Model}
As another example, we analyze the following 2D long-range XX model:
\begin{align}  
	H_{\mathrm{XX}} = \sum_{\substack{\boldsymbol{r}_{1}, \boldsymbol{r}_{2} \\ \boldsymbol{r}_{1} \neq \boldsymbol{r}_{2}}}  
	\frac{1}{2 |\boldsymbol{r}_{1}-\boldsymbol{r}_{2}|^{\alpha}} \br{S_{\boldsymbol{r}_{1}}^{x} S_{\boldsymbol{r}_{2}}^{x} + S_{\boldsymbol{r}_{1}}^{y} S_{\boldsymbol{r}_{2}}^{y}}.
\end{align}  

We calculate the two-point correlation functions $\langle S_{\boldsymbol{i}}^{x} S_{\boldsymbol{i} + \boldsymbol{r}}^{x} \rangle = \mathrm{Tr}(\rho_{\beta} \, S_{\boldsymbol{i}}^{x} S_{\boldsymbol{i} + \boldsymbol{r}}^{x})$ and $\langle S_{\boldsymbol{i}}^{z} S_{\boldsymbol{i} + \boldsymbol{r}}^{z} \rangle = \mathrm{Tr}(\rho_{\beta} \, S_{\boldsymbol{i}}^{z} S_{\boldsymbol{i} + \boldsymbol{r}}^{z})$ for the Gibbs state $\rho_{\beta} = e^{-\beta H_{\mathrm{XX}}} / \mathrm{Tr}(e^{-\beta H_{\mathrm{XX}}})$, between two sites $\boldsymbol{i}$ and $\boldsymbol{i} + \boldsymbol{r}$. The results are illustrated in Fig.~\ref{fig:2D_XX}.

In this scenario, we considered a rectangular lattice with $L_x = 8$ and $L_y = 5$. The correlation functions between $\boldsymbol{i} = (2,3)$ and $\boldsymbol{i} + \boldsymbol{r}$ were computed at $\beta = 0.1$, with $\boldsymbol{r} = (r,0)$ varying for $r$ values ranging from 1 to 4. 
Fig.~\ref{fig:2D_XX}(a) shows $\langle S_{\boldsymbol{i}}^{x} S_{\boldsymbol{i} + \boldsymbol{r}}^{x} \rangle$, while Fig.~\ref{fig:2D_XX}(b) shows $\langle S_{\boldsymbol{i}}^{z} S_{\boldsymbol{i} + \boldsymbol{r}}^{z} \rangle$. In both cases, the results support the clustering property for $\alpha < D$.
As in the previous case, the XTRG algorithm was used to compute the Gibbs state's MPO, employing $\mathrm{U}(1)$ symmetry with a bond dimension of 600.

\begin{figure}[h]
	\centering
	\includegraphics[width=0.7\textwidth]{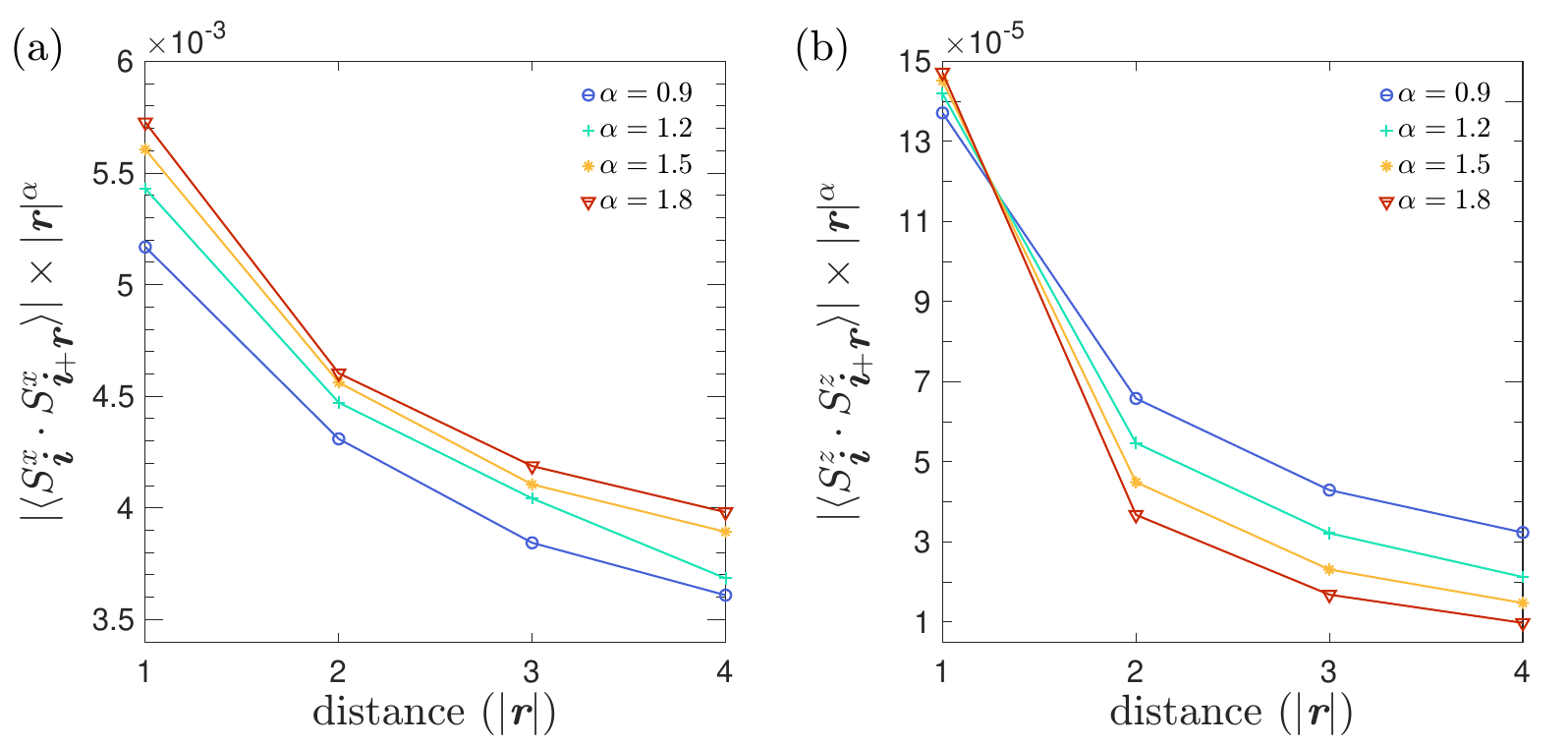}
	\caption{Correlation functions computed for the 2D long-range XX model on a rectangular lattice with $L_x = 8$ and $L_y = 5$, evaluated at $\beta = 0.1$. The calculations were performed for $\boldsymbol{i} = (2,3)$ and $\boldsymbol{r} = (r,0)$ with $r = 1,2,3,4$. 
		(a) The scaled correlation function $\abs{\langle S_{\boldsymbol{i}}^{x} S_{\boldsymbol{i} + \boldsymbol{r}}^{x} \rangle} \times \abs{\boldsymbol{r}}^{\alpha}$ is shown, demonstrating the decay of correlations consistent with the clustering property for $\alpha < D$. 
		(b) The scaled correlation function $\abs{\langle S_{\boldsymbol{i}}^{z} S_{\boldsymbol{i} + \boldsymbol{r}}^{z} \rangle} \times \abs{\boldsymbol{r}}^{\alpha}$ is presented, which similarly supports the clustering property under the same conditions.}
	\label{fig:2D_XX}
\end{figure}

\section{Numerical Detail of Shapourian-Shiozaki-Ryu Negativity}
We consider $n$ spinless fermionic particles with the basis states $|0 \rangle$ and $|1 \rangle = c^{\dagger}|0 \rangle$. The determinant of the $n \times n$ matrix $\mathbf{X} = (x_{i,j})$ is defined as
\begin{align}
	\det \mathbf{X} = \sum_{\sigma \in S_{n}} \mathrm{sgn}(\sigma) x_{1,\sigma(1)} \cdots x_{n,\sigma(n)},
\end{align}
and the Pfaffian of the $2n \times 2n$ skew-symmetric matrix $\mathbf{Y} = (y_{i,j})$, i.e. $\mathbf{Y}^{\mathrm{T}} = - \mathbf{Y}$, is defined as
\begin{align}
	\mathrm{Pf}[\mathbf{Y}] = \frac{1}{2^{n}n!} \sum_{\sigma \in S_{2n}} \mathrm{sgn}(\sigma) y_{\sigma(1) \sigma(2)} \cdots y_{\sigma(2n-1),\sigma(2n)}.
\end{align}
Here, $S_{n}$ and $\mathrm{sgn}(\sigma)$ denotes the symmetric group and the sign of the permutation $\sigma$, respectively.
For an arbitrary $2n \times 2n$ matrix $\mathbf{Z}$, we use the following expression
\begin{align}
	\mathbf{Z} = \begin{pmatrix} [\mathbf{Z}]^{(1,1)} & [\mathbf{Z}]^{(1,2)} \\ [\mathbf{Z}]^{(2,1)} & [\mathbf{Z}]^{(2,2)} \end{pmatrix}
\end{align}
to describe the $n \times n$ block matrix $[\mathbf{Z}]^{(a,b)}$ with $a$-th row and $b$-th column ($a,b \in \{1,2\}$). We denote 
$\mathds{1}_{n}$ as an $n \times n$ identity matrix. 

We use the Grassmann variables $\{\xi_{j},\overline{\xi}_{j}\}$ to describe the fermionic systems with the following notation:
\begin{align}
	\xi &\equiv (\xi_{1},\cdots,\xi_{n}), \\
	(\xi,\overline{\xi}) &\equiv (\xi_{1},\cdots,\xi_{n},\overline{\xi}_{1},\cdots,\overline{\xi}_{n}),
\end{align}
with the differentials
\begin{align}
	d \xi &\equiv d \xi_{n} \cdots d\xi_{1},  \\
	d\overline{\xi} d \xi &\equiv d \overline{\xi}_{n} \cdots d\overline{\xi}_{1} d \xi_{n} \cdots d\xi_{1}, \label{SSR,2} \\
	d (\overline{\xi}, \xi) &\equiv d \overline{\xi}_{1} d \xi_{1} \cdots d \overline{\xi}_{n} d \xi_{n} = (-1)^{n(n-1) / 2} d\overline{\xi} d\xi. \label{SSR,3}
\end{align}
We define the state with the Grassmann variable as
\begin{align}
	|\xi_{j} \rangle &\equiv |0 \rangle_{j} - \xi_{j} |1 \rangle_{j}, \label{SSR,4-1} \\
	|\overline{\xi}_{j} \rangle &\equiv \langle 0 |_{j} - \overline{\xi}_{j} \langle 1 |_{j}, \label{SSR,4-2} \\
	|\xi \rangle \langle \overline{\xi} | &= |\xi_{1} \rangle \langle \overline{\xi}_{1}| \otimes \cdots \otimes |\xi_{n} \rangle \langle \overline{\xi}_{n}|. \label{SSR,5}
\end{align}
The definition naturally leads to $\langle \overline{\xi}_{j} | \xi_{j} \rangle = e^{\overline{\xi}_{j} \xi_{j}}$. We also use the notation $(\chi,\overline{\chi})$ to describe another Grassmann variables.

In terms of the Grassmann variables, the trace of an operator $O$ is given by
\begin{align}
	\mathrm{Tr}[O] = \int \langle - \overline{\xi} | O | \xi \rangle e^{- \sum_{j} \overline{\xi}_{j} \xi_{j}} \, d(\overline{\xi},\xi). \label{SSR,9}
\end{align}
We consider the following form of the density operator:
\begin{align}
	\rho = \frac{1}{Z_{\rho}} \int d\overline{\xi} d \xi \, |\xi \rangle \langle \overline{\xi}| e^{(\xi,\overline{\xi}) \frac{1}{2} \boldsymbol{\Gamma} (\xi,\overline{\xi})^{\mathrm{T}} + \sum_{j} \overline{\xi}_{j} \xi_{j}}. \label{SSR,10}
\end{align}
Imposing $\mathrm{Tr}[\rho] = 1$ with Eqs.~\eqref{SSR,9} and~\eqref{SSR,10}, we have $Z_{\rho} = \mathrm{Pf}[\boldsymbol{\Gamma}]$. Therefore, the density matrix is written as
\begin{align}
	\rho = \frac{1}{\mathrm{Pf}[\boldsymbol{\Gamma}]} \int d\overline{\xi} d \xi \, |\xi \rangle \langle \overline{\xi}| e^{(\xi,\overline{\xi}) \frac{1}{2} \boldsymbol{\Gamma} (\xi,\overline{\xi})^{\mathrm{T}} + \sum_{j} \overline{\xi}_{j} \xi_{j}}. \label{SSR,10-2}
\end{align}
Using Eq.\eqref{SSR,9}, we can relate the two-point correlations functions of the fermionic creation and annihilation operators to the matrix $\boldsymbol{\Gamma}$ by
\begin{align}
	\boldsymbol{\Gamma}^{-1} = \begin{pmatrix} [\boldsymbol{\Gamma}^{-1}]^{(1,1)} & [\boldsymbol{\Gamma}^{-1}]^{(1,2)} \\ [\boldsymbol{\Gamma}^{-1}]^{(2,1)} & [\boldsymbol{\Gamma}^{-1}]^{(2,2)} \end{pmatrix} = \begin{pmatrix} \langle c_{j} c_{i} \rangle_{\rho} & - \langle c_{j}^{\dagger} c_{i} \rangle_{\rho} \\ \langle c_{i}^{\dagger} c_{j} \rangle_{\rho} & \langle c_{j}^{\dagger} c_{i}^{\dagger} \rangle_{\rho} \end{pmatrix}. \label{SSR,41}
\end{align}
In other words, we identify the density matrix with the covaraince matrix $\boldsymbol{\Gamma}$.

Now we move onto the Shapourian-Shiozaki-Ryu (SSR) negativity. We partition the total system $\Lambda$ into $A$ and $B$ with $|A| = n_{A}$ and $|B| = n_{B}$. To define the SSR negativity between $A$ and $B$, we need the partial time-reversal transformation $R_{A}$ on the subsystem $A$. The time-reversal transformation maps $|\xi \rangle \langle \overline{\xi}|$ to $|i \overline{\xi} \rangle \langle i \xi|$ in the coherent basis. Therefore, we have 
\begin{align}
	(|\{\xi_{j}\}_{j \in A}, \{\xi_{j}\}_{j \in B} \rangle \langle \{\overline{\chi}_{j}\}_{j \in A}, \{\overline{\chi}_{j}\}_{j \in B}|)^{R_{A}} = |\{i \overline{\chi}_{j}\}_{j \in A},\{\xi_{j}\}_{j \in B} \rangle \langle \{i \xi_{j}\}_{j \in A},\{\overline{\chi}_{j}\}_{j \in B}|,
\end{align}
and therefore the density matrix after the partial time-reversal transformation becomes
\begin{align}
	\rho^{R_{A}} = \frac{1}{\mathrm{Pf}[\boldsymbol{\Gamma}]} \int d\overline{\xi} \, d\xi \, |i \overline{\xi}_{A}, \xi_{B} \rangle \langle i \xi_{A}, \overline{\xi}_{B} | e^{(\xi,\overline{\xi}) \frac{1}{2} \boldsymbol{\Gamma} (\xi,\overline{\xi})^{\mathrm{T}} + \sum_{j} \overline{\xi}_{j} \xi_{j}}.
\end{align}
Now, we rearrange the Grassmann variables. Let the vector $(\xi',\overline{\xi}') \equiv (\xi_{A}',\xi_{B}',\overline{\xi}_{A}',\overline{\xi}_{B}') = (- i \overline{\xi}_{A}, \xi_{B},-i\xi_{A},\overline{\xi}_{B})$, then we represent it by introducing the matrix $\mathbf{T}$:
\begin{align}
	(\xi',\overline{\xi}')^{\mathrm{T}} = \mathbf{T} (\xi,\overline{\xi})^{\mathrm{T}}, \qquad \mathbf{T} = \begin{pmatrix} 0 & 0 & -i \mathds{1}_{n_{A}} & 0 \\ 0 & \mathds{1}_{n_{B}} & 0 & 0 \\ - i \mathds{1}_{n_{A}} & 0 & 0 & 0 \\ 0 & 0 & 0 & \mathds{1}_{n_{B}} \end{pmatrix}.  \label{SSR,44}
\end{align}
Here, $\mathds{1}_{n_{A}}$ ($\mathds{1}_{n_{B}}$) is an $n_{A} \times n_{A}$ ($n_{B} \times n_{B}$) identity matrix. 
With this matrix, a partial transformed density matrix is rewritten as follows:
\begin{align}
	\rho^{R_{A}} = \frac{1}{\mathrm{Pf}[\boldsymbol{\Gamma}]} \int d\overline{\xi} d \xi \, |\xi \rangle \langle \overline{\xi}| e^{(\xi,\overline{\xi}) \frac{1}{2} \mathbf{S}' (\xi,\overline{\xi})^{\mathrm{T}}}, \label{SSR,45}
\end{align}
with a newly defined matrix
\begin{align}
	\mathbf{S}' := \mathbf{T} \Bigg[\boldsymbol{\Gamma} + \begin{pmatrix} 0 & - \mathds{1}_{n} \\ \mathds{1}_{n} & 0 \end{pmatrix}\Bigg] \mathbf{T}. \label{SSR,46}
\end{align}
The next step is to consider the product $(\rho^{R_{A}})^{\dagger} \rho^{R_{A}}$. We introduce the matrix $\mathbf{S}''$ for the operator $(\rho^{R_{A}})^{\dagger}$:
\begin{align}
	\mathbf{S}'' = \begin{pmatrix} ([\mathbf{S}']^{(2,2)})^{\dagger} & ([\mathbf{S}']^{(1,2)})^{\dagger} \\ ([\mathbf{S}']^{(2,1)})^{\dagger} & ([\mathbf{S}']^{(1,1)})^{\dagger} \end{pmatrix}. \label{SSR,47}
\end{align}
This leads to
\begin{align}
	(\rho^{R_{A}})^{\dagger} &= \frac{1}{\mathrm{Pf}[\boldsymbol{\Gamma}]} \int d\overline{\xi} d\xi \, |\xi \rangle \langle \overline{\xi}| e^{(\xi,\overline{\xi}) \frac{1}{2} (\mathbf{S}'') (\xi,\overline{\xi})^{\mathrm{T}}}.
\end{align}
After a straightforward computation, we find the following expression for the product $(\rho^{R_{A}})^{\dagger} \rho^{R_{A}}$:
\begin{align}
	(\rho^{R_{A}})^{\dagger} \rho^{R_{A}} &= \frac{(-1)^{n^{2}} \mathrm{Pf}[\mathbf{B}]}{\mathrm{Pf}[\boldsymbol{\Gamma}]^{2}} \int d\overline{\xi} d\xi \, |\xi \rangle \langle \overline{\xi}| e^{(\xi,\overline{\xi}) \frac{1}{2} \boldsymbol{\Gamma}' (\xi,\overline{\xi})^{\mathrm{T}} + \sum_{j} \overline{\xi}_{j} \xi_{j}}, \label{SSR,49}
\end{align}
with introducing two matrices
\begin{align}
	\mathbf{B} &= \begin{pmatrix} [\mathbf{S}']^{(1,1)} & - \mathds{1}_{n} \\ \mathds{1}_{n} & [\mathbf{S}'']^{(2,2)} \end{pmatrix}, \label{SSR,52}  \\
	\boldsymbol{\Gamma'} &= - \begin{pmatrix}
		0 & [\mathbf{S}'']^{(1,2)} \\ [\mathbf{S}']^{(2,1)} & 0
	\end{pmatrix} \mathbf{B}^{-1} \begin{pmatrix}
		0 & [\mathbf{S}']^{(1,2)} \\ [\mathbf{S}'']^{(2,1)} & 0
	\end{pmatrix} + \begin{pmatrix}
		[\mathbf{S}'']^{(1,1)} & 0 \\ 0 & [\mathbf{S}']^{(2,2)}
	\end{pmatrix} - \begin{pmatrix}
		0 & - \mathds{1}_{n} \\ \mathds{1}_{n} & 0
	\end{pmatrix}. \label{SSR,51}
\end{align}
Then the SSR negativity is
\begin{align}
	E_{\mathrm{SSR}} = \log \Vert \rho^{R_{A}} \Vert_{1} = \log \mathrm{Tr}\Big[\sqrt{(\rho^{R_{A}})^{\dagger} \rho^{R_{A}}}\Big].
\end{align}

\subsection{Numerical Procedure}
Based on the above calculations, we present the following procedure to numerically calculate the SSR negativity.

\begin{itemize}
	\item[1.] We first compute the covariance matrix $\boldsymbol{\Gamma}$ from Eq.~\eqref{SSR,41}. Then we calculate 
	$\mathbf{S}',\mathbf{S}'',$ and $\mathbf{B}$ from Eqs.~\eqref{SSR,46},~\eqref{SSR,47}, and~\eqref{SSR,52} to construct $\boldsymbol{\Gamma}'$ by Eq.~\eqref{SSR,51}.
	\item[2.] Note the relation
	\begin{align}
		(\rho^{R_{A}})^{\dagger} \rho^{R_{A}} &= \frac{(-1)^{n^{2}} \mathrm{Pf}[\mathbf{B}] \mathrm{Pf}[\boldsymbol{\Gamma}']}{\mathrm{Pf}[\boldsymbol{\Gamma}]^{2}} \rho'', \\
		\rho'' &= \frac{1}{\mathrm{Pf}[\boldsymbol{\Gamma}']} \int d\overline{\xi} d\xi \, |\xi \rangle \langle \overline{\xi}| e^{(\xi,\overline{\xi}) \frac{1}{2} \boldsymbol{\Gamma}' (\xi,\overline{\xi})^{\mathrm{T}} + \sum_{j} \overline{\xi}_{j} \xi_{j}}.
	\end{align}
	From $\boldsymbol{\Gamma}'$, using Eq.~\eqref{SSR,41} again, we define the covariance matrix $\mathbf{C}$ of $\rho''$:
	\begin{align}
		\mathbf{C} = \begin{pmatrix} \langle c_{i}^{\dagger} c_{j} \rangle_{\rho''} & \langle c_{i}^{\dagger} c_{j}^{\dagger} \rangle_{\rho''} \\ \langle c_{i} c_{j} \rangle_{\rho''} & \langle c_{i} c_{j}^{\dagger} \rangle_{\rho''} \end{pmatrix} = \begin{pmatrix} [(\boldsymbol{\Gamma}')^{-1}]^{(2,1)} & \Big([(\boldsymbol{\Gamma}')^{-1}]^{(2,2)}\Big)^{\mathrm{T}} \\ \Big([(\boldsymbol{\Gamma}')^{-1}]^{(1,1)}\Big)^{\mathrm{T}} & \mathds{1}_{n} + [(\boldsymbol{\Gamma}')^{-1}]^{(1,2)} \end{pmatrix}.
	\end{align}
	\item[3.] We choose the unitary transformation $U$ to define another fermionic annihilation operators $\mathbf{d} = (d_{1},\cdots,d_{n})$:
	\begin{align}
		\begin{pmatrix} \mathbf{d}^{\dagger} \\ \mathbf{d} \end{pmatrix} = U \begin{pmatrix} \mathbf{c}^{\dagger} \\ \mathbf{c} \end{pmatrix},
	\end{align}
	so that it diagonalizes $\rho''$ as
	\begin{align}
		\rho'' = \frac{e^{-\sum_{i = 1}^{n} \epsilon_{i}'' d_{i}^{\dagger} d_{i}}}{Z''}, \label{Conditiond}
	\end{align}
	where $Z'' = \mathrm{Tr}[\rho''] = \prod_{i = 1}^{n} (1 + e^{-\epsilon_{i}''})$. Equation.~\eqref{Conditiond} yields
	\begin{align}
		\begin{pmatrix} \langle d_{i}^{\dagger} d_{j} \rangle_{\rho''} &  \langle d_{i}^{\dagger} d_{j}^{\dagger} \rangle_{\rho''} \\ \langle d_{i} d_{j} \rangle_{\rho''} & \langle d_{i} d_{j}^{\dagger} \rangle_{\rho''} \end{pmatrix} = \begin{pmatrix}
			\frac{1}{e^{\epsilon_{i}''} + 1} \delta_{i,j} & 0 \\ 0 & \frac{e^{\epsilon_{i}''}}{e^{\epsilon_{i}''} + 1} \delta_{i,j}
		\end{pmatrix} = U \begin{pmatrix} \langle c_{i}^{\dagger} c_{j} \rangle_{\rho''} &  \langle c_{i}^{\dagger} c_{j}^{\dagger} \rangle_{\rho''} \\ \langle c_{i} c_{j} \rangle_{\rho''} & \langle c_{i} c_{j}^{\dagger} \rangle_{\rho''} \end{pmatrix} U^{\dagger}.
	\end{align}
	It means that we can obtain $\epsilon_{i}''$ by diagonalizing the covariance matrix $\mathbf{C}$:
	\begin{align}
		\mathbf{C} &= U^{\dagger} \mathrm{diag}(\langle d_{i}^{\dagger}d _{i} \rangle_{\rho''}, \langle d_{i} d_{i}^{\dagger} \rangle_{\rho''}) U, \\
		\langle d_{i}^{\dagger} d_{i} \rangle_{\rho''} &= \frac{1}{e^{\epsilon_{i}''} + 1}.
	\end{align}
	\item[4.] From the following identity
	\begin{align}
		\mathrm{Tr}\Big[\sqrt{(\rho^{R_{A}})^{\dagger} \rho^{R_{A}}}\Big] &= \sqrt{\frac{(-1)^{n^{2}} \mathrm{Pf}[\mathbf{B}] \mathrm{Pf}[\boldsymbol{\Gamma}']}{\mathrm{Pf}[\boldsymbol{\Gamma}]^{2}}} \mathrm{Tr}\sqrt{\rho''} = \sqrt{\frac{\mathrm{Pf}[\mathbf{B}] \mathrm{Pf}[\boldsymbol{\Gamma}']}{\mathrm{Pf}[\boldsymbol{\Gamma}]^{2}}} \prod_{i = 1}^{n} \frac{1 + e^{-\frac{\epsilon_{i}''}{2}}}{\sqrt{1 + e^{-\epsilon_{i}''}}},
	\end{align}
	we get the following expression of the SSR negativity with computable quantities:
	\begin{align}
		E_{\mathrm{SSR}} &= \log \mathrm{Tr}\Big[\sqrt{(\rho^{R_{A}})^{\dagger} \rho^{R_{A}}}\Big] \\
		&= \frac{1}{2} \log \Bigg[\frac{(-1)^{n^{2}} \mathrm{Pf}[\mathbf{B}] \mathrm{Pf}[\boldsymbol{\Gamma}']}{\mathrm{Pf}[\boldsymbol{\Gamma}]^{2}}\Bigg] + \sum_{i=1}^{n} \Bigg[\log \Big(1 + e^{- \frac{\epsilon_{i}''}{2}}\Big) - \frac{1}{2} \log (1 + e^{-\epsilon_{i}''})\Bigg] \label{SSR,60} \\
		&= \frac{1}{2} \log \Bigg[\frac{(-1)^{n^{2}} \mathrm{Pf}[\mathbf{B}] \mathrm{Pf}[\boldsymbol{\Gamma}']}{\mathrm{Pf}[\boldsymbol{\Gamma}]^{2}}\Bigg] + \sum_{i=1}^{n} \log \Big(\sqrt{\langle d_{i}^{\dagger} d_{i} \rangle_{\rho''}} + \sqrt{\langle d_{i} d_{i}^{\dagger} \rangle_{\rho''}}\Big). \label{SSR,61}
	\end{align}
\end{itemize}

%\section{Methods}
%
%

\end{widetext}

\end{document}